\documentclass[aps,notitlepage,twocolumn,nofootinbib,prx,10pt]{revtex4-2}
\usepackage{amsmath,amsfonts,amssymb,amsthm,mathtools,bbm,graphicx,enumerate,dsfont,times}

\usepackage[
    breaklinks,
    pdftex,
    colorlinks=true,
    linkcolor=teal,
    citecolor=teal,
    urlcolor=teal
]{hyperref}

\usepackage{tikz-cd}
\usepackage{stmaryrd}
\usepackage{quantikz}

\usepackage[para]{threeparttable}
\usepackage[all]{xy}
\usepackage[normalem]{ulem}
\usepackage{physics}
\usepackage{algorithm}
\usepackage{algorithmic}
\usepackage{xcolor} 
\newtheorem{thm}{Theorem}

\newtheorem{cor}[thm]{Corollary}
\newtheorem{lem}[thm]{Lemma}

\newtheorem{dfn}[thm]{Definition}

\newtheorem{con}[thm]{Conjecture}

\def\tr{\operatorname{tr}}   
\def\ii{\mathrm{i}}

\begin{document}

\title{Emergent statistical mechanics in holographic random tensor networks}
\author{Shozab Qasim, Jens Eisert, and Alexander Jahn}
\affiliation{Department of Physics, Freie Universit\"at Berlin, 14195 Berlin, Germany}

\begin{abstract}
Recent years have enjoyed substantial progress in capturing properties of complex quantum systems by means of random tensor networks (RTNs),
which form ensembles of quantum states that depend only on the tensor network geometry and bond dimensions. Of particular interest are RTNs on hyperbolic geometries, with local tensors typically chosen from the unitary Haar measure, that model critical boundary states of holographic bulk-boundary dualities. In this work, we elevate static pictures of ensemble averages to a dynamical one, to show that RTN states exhibit equilibration of time-averaged operator expectation values under a highly generic class of Hamiltonians with non-degenerate spectra. 
We prove that RTN states generally equilibrate at large bond dimension and also in the scaling limit for three classes of geometries: Those of matrix product states, regular hyperbolic tilings, and single ``black hole'' tensors.
Furthermore, we prove a hierarchy of equilibration between finite-dimensional instances of these classes for bulk and boundary states with small entanglement. This suggests an equivalent hierarchy between corresponding many-body phases, and reproduces a holographic degree-of-freedom counting for the effective dimension of each system. These results demonstrate that RTN techniques can probe aspects of late-time dynamics of quantum many-body phases and suggest a new approach to describing aspects of holographic dualities using techniques from statistical mechanics.
\end{abstract}

\date{\today}
\maketitle

\section{Introduction}

Complex quantum systems encountered in condensed-matter and high-energy physics exhibit a rich array of fascinating phenomena. While their behavior is compelling, these systems are also notoriously difficult to describe theoretically. To address this challenge, physicists have long relied on simplified ``proxies'' that, while easier to analyze, still capture key features of the underlying physics. Among the most powerful of these tools are tensor networks and quantum circuits: 
models that, despite being governed by relatively few parameters, can accurately represent a wide range of physical scenarios \cite{Fannes:1990ur, Verstraete_2008, Cirac_2021}.
Along similar lines, recent years have seen \emph{randomness}  emerge as a particularly valuable ingredient in this context.
Random quantum circuits \cite{Fisher_2023, Emerstmruc, Nahum_2018, Nahum_2017, hunterjones2019unitarydesignsstatisticalmechanics, Brand_o_2019, Haferkamp_2022, Tan_2025, Vasseur_2019, Fava_2025}, for example, serve 
as effective and versatile models for chaotic quantum many-body dynamics \cite{Sekino:2008he,Shenker_2014, shenker2015stringyeffectsscrambling, Maldacena_2016, Cotler_2017, Gharibyan_2018, saad2019semiclassicalrampsykgravity}, while \emph{random tensor networks} (RTNs) have proven successful in reproducing static 
features of holographic dualities, notably in the context of 
the AdS/CFT correspondence
\cite{Hayden:2016cfa,jia2020petzreconstructionrandomtensor,Chandra_2023,cheng2022randomtensornetworksnontrivial}, a bulk/boundary duality that relates quantum gravity in an asymptotically \emph{anti-de Sitter}  (AdS) bulk spacetime to a \emph{conformal field theory} (CFT) on the spacetime boundary  \cite{Maldacena_1999, witten1998antisitterspaceholography}.
A crucial feature of this setting is the emergence of a bulk spacetime whose connectivity is closely related to the entanglement structure of the boundary CFT, similar to how a hyperbolic tensor network geometry represents its ground state entanglement \cite{Swingle:2009bg,Nozaki:2012zj}. 

Thus far, RTNs -- a simple example of which is shown in Fig.~\ref{fig:equilibration-intro}(a) -- have been shown to capture key features expected of a holographic CFT, displaying an analog of the 
\emph{Ryu-Takayanagi} (RT) formula \cite{Ryu_2006, Ryu_2006long} for R\'enyi entropies in the large bond dimension limit. According to the RT formula, the entanglement entropy for a CFT subregion $A$ is proportional to the area of a matching minimal surface $\gamma_A$ in the bulk geometry, and a similar scaling appears in RTNs. 
They have also been shown to display algebraically decaying correlations as expected by CFT states.
Furthermore, they reproduce a gap in the operator spectrum that is typically expected of the ``gapped CFTs'' appearing in holography \cite{Heemskerk_2009, El_Showk_2012}. Holographic random tensor networks in the large bond dimension limit have also been associated with fixed-area states in quantum gravity, as they demonstrate a corresponding flatness in their entanglement spectrum \cite{Akers:2018fow,Dong_2019}, though variants with non-flat entanglement spectra can also be constructed by adding non-maximally entangled \emph{link states} \cite{cheng2022randomtensornetworksnontrivial}. 
Holographic tensor network models, however, have so far not been shown to model the thermalization and equilibration behavior of holographic CFTs \cite{lashkari2016eigenstatethermalizationhypothesisconformal, saad2019latetimecorrelationfunctions, Pollack_2020, Bao_2019, jafferis2023jt, jafferis2023matrix, Sonner_2017, Nayak_2019, Dymarsky_2018, Dymarsky:2016ntg, Basu_2017, Lashkari:2017hwq, Faulkner:2017hll, Brehm:2018ipf, Romero-Bermudez:2018dim}.

In this work, we take steps towards exploring the 
\emph{non-equilibrium dynamics}  \cite{Gogolin_2016,PolkovnikovReview,Eisert_2015} of holographic toy models constructed from RTNs, making a precise and concrete step to resolve this issue. 
As random tensor networks produce an ensemble of states with no clear Hamiltonian to describe their time evolution, it had generally been assumed that they are severely limited in modeling dynamical phenomena of holographic boundary theories.
Nonetheless, building upon the results presented in Ref.\ \cite{Haferkamp_2021}, we show that the time dynamics of random tensor network states under highly generic Hamiltonians are in fact computable. 
We specifically model the phenomenon of equilibration, defined in terms of the deviations of the expectation value of simple operators around the thermal equilibrium value at infinite time.
The randomness inherent to RTNs allow us to make rigorous statements about equilibration whilst respecting the locality structure of $(1{+}1)$-dimensional boundary theories, beyond what is probed, e.g., by random matrix models \cite{Hubener:2014pfa}. 

\begin{figure*}[t]
    \centering
    \includegraphics[width=0.9\linewidth]{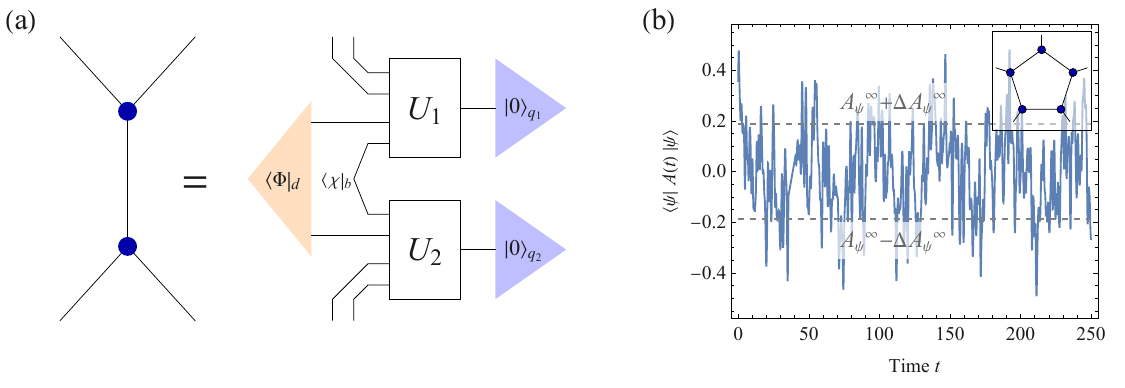}
    \caption{(a) Construction of a random tensor network (RTN) state vector $\ket{\psi}$ on a graph: Each vertex $k$ is associated with a $q_k$-dimensional Haar-random unitary $U_k$ acting on a reference state vector $\ket{0}_{q_k}$, and each closed edge with projection on an EPR pair $\ket{\chi}_b$ of bond dimension $b$. Each vertex may include ``bulk'' degrees that are jointly projected onto a state vector $\ket{\Phi}_d$ of dimension $d$. The ``boundary'' state vector $\ket{\psi}$ is identified with the open edges. 
    (b) Equilibration of an observable 
    $A(t) = e^{\ii H t} A e^{-\ii H t}$ is 
    determined by the late-time fluctuations $\Delta A_\psi^\infty$ around the average $A_\psi^\infty$. Here we show an example of an RTN given by a five-spin \emph{matrix product state} (MPS) under an Ising Hamiltonian $H$ and $A$ as Pauli $X$ acting on the first site. Overlaid are the analytical deviations $\Delta A_\psi^\infty \approx 0.188$ from \eqref{EQ_EFF_DIM_MPS_CLOSED}, calculated independently from $H$ but tightly upper-bounding the standard deviation $\sigma_{\langle A \rangle} \approx 0.175$ of samples of expectation values $\langle A \rangle$.
    }
    \label{fig:equilibration-intro}
\end{figure*}

We will begin with a primer on ensembles, typicality, and equilibration in Sec.~\ref{SEC_BACKGROUND}, and then discuss our setup and basic calculation tool (Weingarten calculus) in Sec.~\ref{SEC_SETUP}. 
The main results will be presented in Sec.~\ref{SEC_RESULTS}. 
There we will show that perfect equilibration -- defined as zero fluctuations of local observables at late times -- appears in RTNs in two limits: One is the \emph{continuum limit} where the local dimension of boundary sites diverges, and the other is the \emph{scaling limit} where the geometry of the tensor network and the number of its boundary sites becomes infinitely large. We will prove perfect equilibration in the continuum limit for generic RTNs, but only prove it for the scaling limit for three types of RTN geometries: Linear chains commonly known as \emph{matrix product states} (MPS) or \emph{tensor train states} (TTS), non-triangular regular hyperbolic tilings (typically used as holographic models), and single random tensors.
We will also be able to distinguish equilibration between finite-dimensional instances of these three RTN geometries for a boundary system of fixed size and dimension, establishing a hierarchy of equilibration in which regular hyperbolic RTN states provably equilibrate more strongly than RTN states from MPS geometries, but equilibrate more weakly than single random tensor states.
For generic RTNs, we will further prove that equilibration of an RTN can never be decreased by \emph{fusing} two connected tensors into a single larger random tensor, thus relating equilibration between more general RTN geometries. We additionally extend our results to RTN's constructed from all of Dysons circular ensembles and finally make an extension to proving equilibration for multi-point observables. 
The implications of our results for many-body physics as well as holography will then be discussed in Sec.~\ref{SEC_DISCUSSION}.

\section{Background}
\label{SEC_BACKGROUND}

In this section, we set the stage for applying notions of \emph{pure-state statistical mechanics} \cite{Gogolin_2016, Linden_2009} to our setting at hand. We formally elucidate these aspects in Appendix \ref{APP:DETEQUIB}. 
Statistical mechanics is formulated in terms of ensembles. The micro-canonical ensemble in quantum statistical mechanics takes the form of the micro-canonical state (formally defined in \eqref{microcanonical}). Both the micro-canonical ensemble and state are defined with respect to an energy window, with the micro-canonical state encoding the uniform mixture over eigenstates in that energy window. However, the applicability of an ensemble is a postulate. In order to derive the rules of statistical mechanics, one needs a deeper argument. An example of such an argument is the notion of \emph{typicality} which can be made rigorous by showing that most states in the ensemble will be indistinguishable by low-complexity (local or coarse-grained) observables.
This can be made concrete by demonstrating a \emph{concentration of measure} phenomenon. A natural way to formulate such arguments is in terms of state vectors for high dimensional subspaces which have reduced states on small subsystems that appear close to the reduced state of the micro-canonical state corresponding to that subspace. The likelihood of picking these state vector from this subspace is equal. 

A most natural measure for this is the \emph{Haar measure} \cite{haar1933,Weingarten:1977ya} for which one can show that for any bounded observable $A$, the likelihood that the expectation value of the observable for a micro-canonical state deviating 
with respect to a state picked from the Haar measure is exponentially small in the dimension of the Hilbert space of the energy window. 
In a lattice system, the dimension of the Hilbert space in the energy window that is far from the spectral edges (and within a fixed superselection sector) will typically grow exponentially with system size. 

Consequently, considering the role of measurements also enforces an indistinguishability of subsystems. In other words, restricting to low-complexity measurements (with support on subsystems that are much smaller than the dimension of the energy window) will lead to an indistinguishability between the subsystem density matrix of the Haar-random state and the subsystem microcanonical density matrix - in other words, they will be exponentially close in trace distance (Theorem \ref{measureconcentration}). In this work, we will be concerned with equilibration of such simple observables, i.e., few-body local operators. 

Thus far we have fixed an energy window and described typicality, as a static statement - however, we can make a dynamical statement, namely concerning equilibration. Under chaotic unitary dynamics, a system will scramble. If the energy levels of a system are sufficiently populated, the rapid phase mixing will mean that expectation value of an observable will, for most times, be near a steady value that will be obtained by time-averaging. The time averaged state will be determined by the energy level population. Generic states will be those for which energy level population will be nearly uniform. Typicality essentially guarantees that most pure states in an energy window will have almost uniform populations - thus, their time-averaged state will locally behave as the microcanonical state (Theorem \ref{equibonav}). This also holds physical relevance as most real experiments do not resolve exponentially many outcomes. 

However, even for simple measurements, any two local pure states which differ microscopically will be indistinguishable from the microcanonical state. Thus all one observes in practice is equilibration. Proofs of \emph{thermalization}, require a much stronger condition than equilibration, can be achieved using random matrix theory \cite{weidenmüller2022randommatrixmodelthermalization}. 
However, due to the randomness, the long-time expectation values of observables become completely independent of the initial state and  the initial energy.
Eigenstates of physical Hamiltonians have additional structure beyond a random matrix \cite{Keating_2015}.

The structure present is elucidated by a reading of the so-called \emph{eigenstate thermalisation hypothesis}
(ETH) \cite{Srednicki_1994, RevModPhys.91.021001, D_Alessio_2016} and its generalization \cite{Pappalardi_2022, Foini_2019}. An alternative route which we spearhead in this work is to exploit unitary random matrices to construct states for which we prove equilibration. As a result, we can prove equilibration for random states exhibiting locality -- and yet, exploit techniques from random matrices in our rigorous proofs. We will make use of the following lemma,
expressed in terms of the \emph{effective dimension} $D_\text{eff} 
\coloneqq D_\text{eff}(\ket{\psi})$ 
\cite{Linden_2009,Gogolin_2016},
which is for a state vector $\ket\psi$ defined as
\begin{equation}
\label{EQ_EFF_DIM_DEF}
    \frac{1}{D_\text{eff}} \coloneqq \sum_j |\langle \psi | j \rangle|^4 \ ,
\end{equation}
with the $\ket{j}$ forming the energy eigenbasis of $H$ which captures how many eigenstates of the Hamiltonian substantially contribute to a given state vector.

\begin{lem}[\cite{huang2020instabilitylocalizationtranslationinvariantsystems}]\label{fluclem}
    For any Hamiltonian $H$ whose spectrum has non-degenerate gaps and any operator \( A \) with bounded norm \( \|A\| = O(1) \),
\begin{equation}
    (\Delta A_{\psi}^{\infty})^2 = O(1/D_\text{eff}).
\end{equation}
\end{lem}

\begin{proof}
For clarity, we reproduce a proof of this lemma from Ref.~\cite{huang2020instabilitylocalizationtranslationinvariantsystems}, but with a modified definition of $(\Delta A_{\psi}^{\infty})^2$ to coincide with the typical definition of a variance.
Let \( c_j := \langle \psi | j \rangle \) be the coefficients of the state vector expressed in the Hamiltonian eigenbasis.
Writing out the matrix elements, one finds
\begin{align}
    (\Delta A_{\psi}^{\infty})^2 &= \lim_{\tau \to \infty} \frac{1}{\tau} \int_0^{\tau} \left| \langle \psi | A(t) | \psi \rangle - A_{\psi}^{\infty} \right|^2 dt \\
    &= \lim_{\tau \to \infty} \frac{1}{\tau} \int_0^{\tau} \left| \sum_{j \neq k} c_j c_k^* A_{j,k} e^{i(E_j - E_k)t} \right|^2 dt \nonumber\\
    &= \sum_{j \neq k, j' \neq k'} c_j c_k^* c_{j'} c_{k'}^* A_{j,k} A_{k'j'}^* \nonumber\\ \nonumber
    &\times \lim_{\tau \to \infty} \frac{1}{\tau} \int_0^{\tau} e^{i(E_j - E_k - E_{j'} + E_{k'})t} dt \nonumber\\
    &= \sum_{j \neq k} |c_j|^2 |c_k|^2 A_{j,k} A_{k,j}^{\dagger} ,
\end{align}
which gives
    \begin{align}
    (\Delta A_{\psi}^{\infty})^2 
    &\leq \sqrt{\sum_{j \neq k} |c_j|^4 A_{j,k} (A^{\dagger})_{k,j}  \times \sum_{j \neq k} |c_k|^4 (A^{\dagger})_{k,j} A_{j,k} } \nonumber\\
    &\leq \sqrt{\sum_j |c_j|^4 (A A^{\dagger})_{j,j}  \times \sum_k |c_k|^4 (A^{\dagger} A)_{k,k} } \nonumber\\
    &= \sum_j O(|c_j|^4) = O(1/D_\text{eff}),
\end{align}
where we have made use of the assumption that the spectrum of the Hamiltonian has non-degenerate gaps and the last step relies on 
the pinching inequality
\begin{equation}
(AA^{\dagger})_{j,j} \leq 
\|A\|^2 = O(1) .
\end{equation}
\end{proof} 
Intuitively, the condition of non-degenerate gaps ensures that over sufficiently long times, constructive interference of the phases for different energy eigenstate contributions is strongly suppressed. This ensures that deviations from equilibrium are for most times contained within a window of states of size $D_\text{eff}$.
We show an example of the equilibration of an RTN state and fluctuations $\Delta A_\psi^\infty$ around a late-time average $A_\psi^\infty$ for a simple Pauli observable $A=X_1$ acting on a single site in Fig.~\ref{fig:equilibration-intro}(b).

The inverse effective dimension is additionally related to \emph{delocalization} via the \emph{inverse participation ratio} (IPR) and the Loschmidt echo \cite{Gogolin_2016}. Given a general quantum state vector 
\begin{equation}
\ket{\psi}=\sum_i c_i \ket{\psi}, 
\end{equation}
the IPR is given by 
$\sum_i|c_i|^4$. This quantity measures how concentrated a state vector 
is in the energy eigenbasis $\{\ket{i}\}$ (note that this can be computed in real space or momentum space). Anderson localized states have an IPR close to $1$ while delocalized states (ergodic or thermal) states have a smaller IPR. The Loschmidt echo is defined given an initial state vector $\ket{\psi}$ and two Hamiltonians $H$ and $H'=H+\delta H$, the 
\emph{Loschmidt echo} is defined as 
\begin{equation}
\mathcal{L}(t)\coloneqq |\bra{\psi}e^{iH't}e^{-iHt}\ket{\psi}|^2. 
\end{equation}
The infinite time average of the Loschmidt echo is then given by
\begin{equation}
\overline{\mathcal{L}}=\lim_{T\xrightarrow{}\infty}\int_0^T \mathcal{L}(t). 
\end{equation}
Expanding in the energy eigenbasis of the Hamiltonian $H$, we have $\ket{\psi}=\sum_n c_n \ket{E_n}$ and thus 
\begin{align}
    \mathcal{L}(t) = |\sum_n |c_n|^2e^{-iE_nt}|^2=\sum_{n,m} |c_n|^2|c_m|^2 e^{-i(E_n-E_m)t}
\end{align}
so that $\overline{\mathcal{L}}=\sum_n|c_n|^4$ and the Loschmidt echo will be exponentially small for a delocalized quantum state.

\section{General setup}
\label{SEC_SETUP}

\subsection{Random matrix techniques}

We define a \emph{random tensor network} (RTN) analogous to the construction in Ref.~\cite{Hayden_2016}, where each random tensor of total dimension $D$ (i.e., the dimension of all its legs) is constructed by acting with a random unitary $U$ 
on a reference state vector $\ket{0}_D$ (see Fig.~\ref{fig:equilibration-intro}(a) for an example with two tensors). 
The space of unitary matrices is special because it forms a compact group. There are well established ways to specify and determine a volume form which defines a uniform measure on this space. The volume form, or Haar form, in keeping with the terminology of the corresponding measure which is called the Haar measure, is denoted by $\mathrm{d}U$ and is required to satisfy homogenity, in that for any fixed matrix $V$, 
\begin{equation}
\label{volform}
    \mathrm{d}UV = \mathrm{d}VU = \mathrm{d}U
\end{equation}
which is referred to as invariance under left and right actions of the group. In this sense, the Haar form implies a uniform measure on the space of unitary matrices. The Haar form exists and is unique up to normalization for all compact groups. 
We can now make an explicit definition of an ensemble of random unitary matrices as 
follows.

\begin{dfn}[Circular unitary ensemble \cite{forrester2010loggases}]
    The circular unitary ensemble (CUE) is the group of unitary matrices endowed with the volume form (\ref{volform}).
\end{dfn}

We will be interested in the integration of a polynomial function $f(U)$ of the matrix elements of an $N \times N$ unitary matrix $U$ over the unitary group $U(N)$. The fundamental object of consideration is 
\begin{equation}
    \int \mathrm{d}U U_{i_1 j_1} \dots U_{i_k j_k} U^\dagger_{l_1 m_1} \dots U^\dagger_{l_k m_k} .
\end{equation}
We now arrive at the following result.

\begin{thm}[Weingarten calculus \cite{Collins_2006, köstenberger2021weingartencalculus, Brouwer_1996}] 
\begin{align}
    &\int \mathrm{d}U U_{i_1 ,j_1} \dots U_{i_k ,j_k} U^\dagger_{l_1 ,m_1} \dots U^\dagger_{l_k, m_k} = \nonumber \\ 
    &\sum_{\sigma, \tau \in S_k} \delta_{\sigma} (\mathbf{i}|\mathbf{m}) \delta_\tau (\mathbf{j}|\mathbf{l})\text{Wg}(\sigma \tau^{-1},d)
\end{align}
where we sum over elements of the permutation group $S_k$ and define a contraction of indices with respect to a permutation $\sigma \in S_k$ as 
\begin{equation}
\delta_\sigma(\mathbf{i}|\mathbf{j}) := \prod_{s=1}^{k} \delta_{i_s j_{\sigma(s)}} = 
     \delta_{i_1, j_{\sigma(1)}} \dots \delta_{i_k ,j_{\sigma(k)}}
\end{equation}
averages of $U^{\otimes k} \otimes U^{\dagger \otimes k'}$ with $k \neq k'$ vanish identically. The index contraction $\delta_{\sigma}(\mathbf{i}|\mathbf{j})$ can be interpreted as a permutation operator that acts on the basis of the $k-$fold space as
\begin{equation}
\delta_\sigma(\mathbf{i}|\mathbf{j}) = P_\sigma
\end{equation}
The weight associated to a given contraction is called the Weingarten function. It is a function on elements of $S_k$ and admits an expansion in terms of characters of the symmetric group 
\begin{equation}
\label{weingarten}
    \text{Wg}(\sigma,d) = \frac{1}{k\!} \sum_{\lambda} \frac{f_\lambda \chi_\lambda(\sigma)}{c_\lambda (d)}
\end{equation}
where we sum over integer partitions of $k$ which label the irreps of $S_k$. $\chi_\lambda(\sigma)$ is an irreducible character of $\lambda$, and $f_\lambda$ is the dimension of the irrep $\lambda$. The polynomial in the denominator is defined as 
\begin{equation}
    c_\lambda(d) = \prod_{(i,j) \in \lambda}(d+j-1),
\end{equation}
where we take a product over the coordinates $(i,j)$ of the Young diagram of $\lambda$. Writing $\lambda$ as an integer partition of $k$, with elements $\lambda_i$, the product is taken over $i$ from $1$ to $l(\lambda)$, the length of the partition, and $j$ from $1$ to $\lambda_i$. The expression for the Weingarten function (\ref{weingarten}), is valid for $k \geq d$ by restricting the sum over partitions of length $l(\lambda) \leq d$ so that there are no poles in the polynomial $c_\lambda(d)$.
\end{thm}
The instances we make use of are
\begin{align}
\label{weingarten2n4a}
    & \int \mathrm{d}U U_{i_1,j_1}U^*_{i_2,j_2} = \frac{1}{D}\delta_{i_1,i_2}\delta_{j_1,j_2}, \\
\label{weingarten2n4b}
    & \int \mathrm{d}U U_{i_1,j_1}U_{i_2,j_2}U_{i_3,j_3}^*U_{i_4,j_4}^* =  \\
     &\frac{1}{D^2-1}\left[\delta_{i_1,i_3}\delta_{i_2,i_4}\delta_{j_1,j_3}\delta_{j_2,j_4} - \frac{1}{D}\delta_{i_1,i_3}\delta_{i_2,i_4}\delta_{j_1,j_4}\delta_{j_2,j_3} \right] \nonumber \\
     &+ \frac{1}{D^2-1}\left[- \frac{1}{D} \delta_{i_1,i_4}\delta_{i_2,i_3}\delta_{j_1,j_3}\delta_{j_2,j_4} +\delta_{i_1,i_4}\delta_{i_2,i_3}\delta_{j_1,j_4}\delta_{j_2,j_3} \right] .\nonumber
\end{align}
The corresponding graphical rules for these integrals are given by 
\begin{align}
\label{EQ_WEINGARTEN_STATEBASE1}
    &\begin{gathered}
         \includegraphics[width=0.32\linewidth]{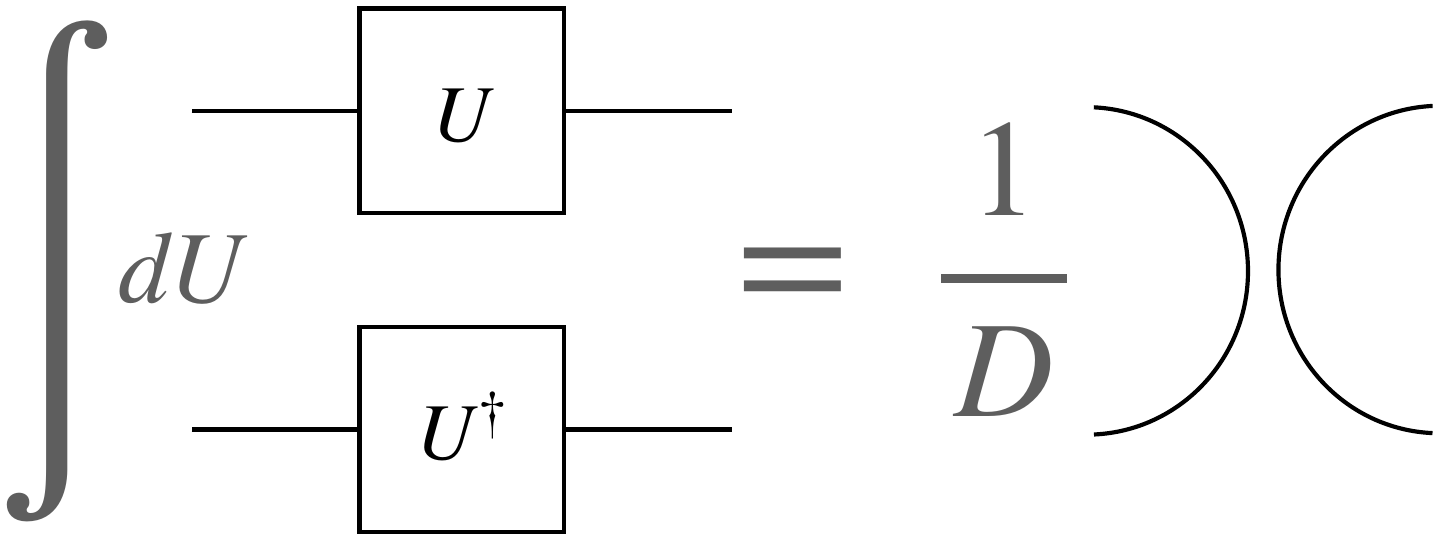}
    \end{gathered} \ , \\[3pt]
\label{EQ_WEINGARTEN_STATEBASE2}
    & \begin{gathered}
         \includegraphics[width=0.85\linewidth]{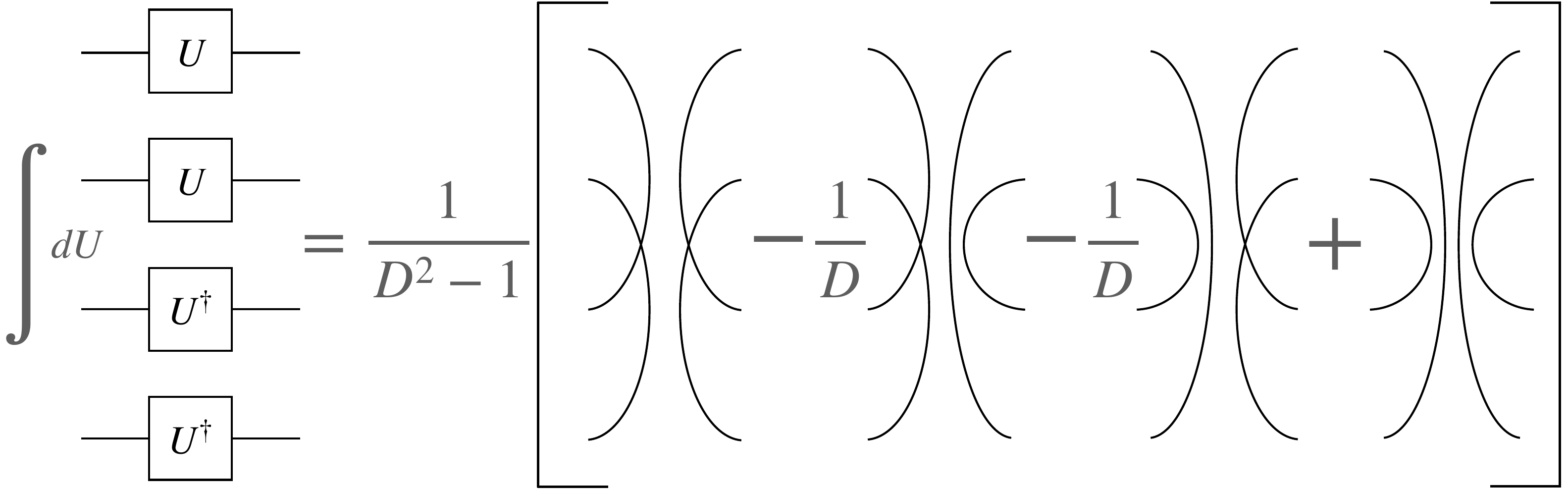}
    \end{gathered} \ .
\end{align}
Acting with these expressions on any reference state, i.e., computing the expectation value of a product of Haar-random \emph{states}, leads to the 
graphical rules
\begin{align}
\label{EQ_WEINGARTEN_STATE_2}
    &\begin{gathered}
         \includegraphics[width=0.45\linewidth]{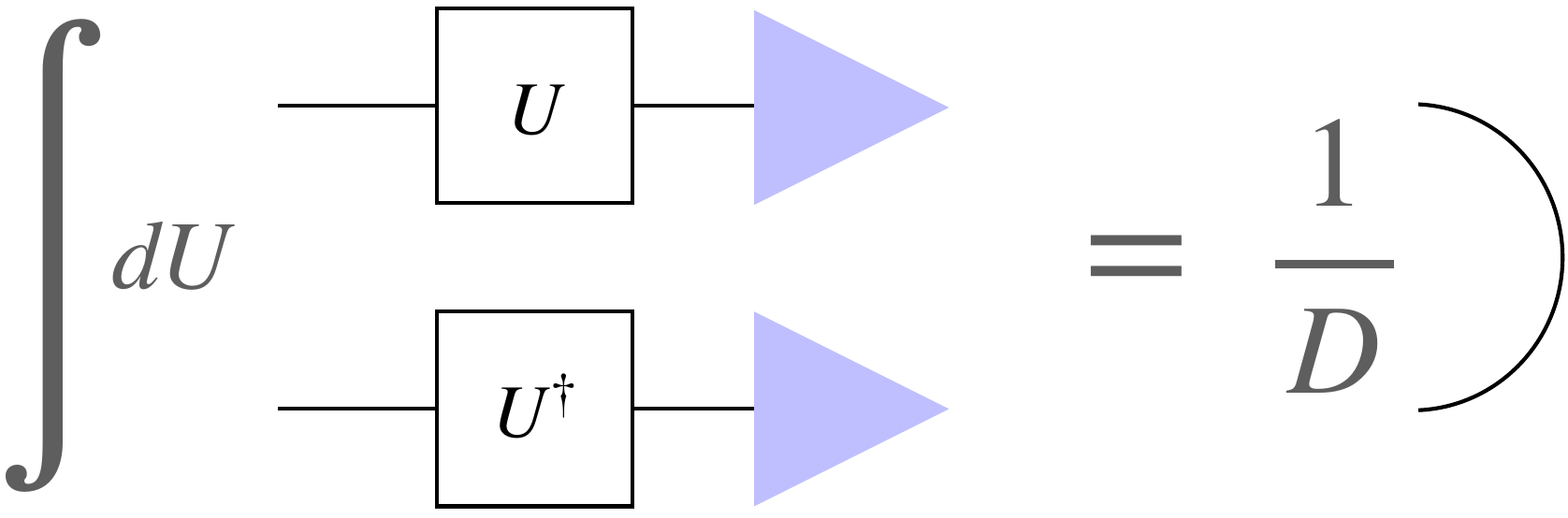}
    \end{gathered} \ , \\[4pt]
\label{EQ_WEINGARTEN_STATE_4}
    & \begin{gathered}
         \includegraphics[width=0.85\linewidth]{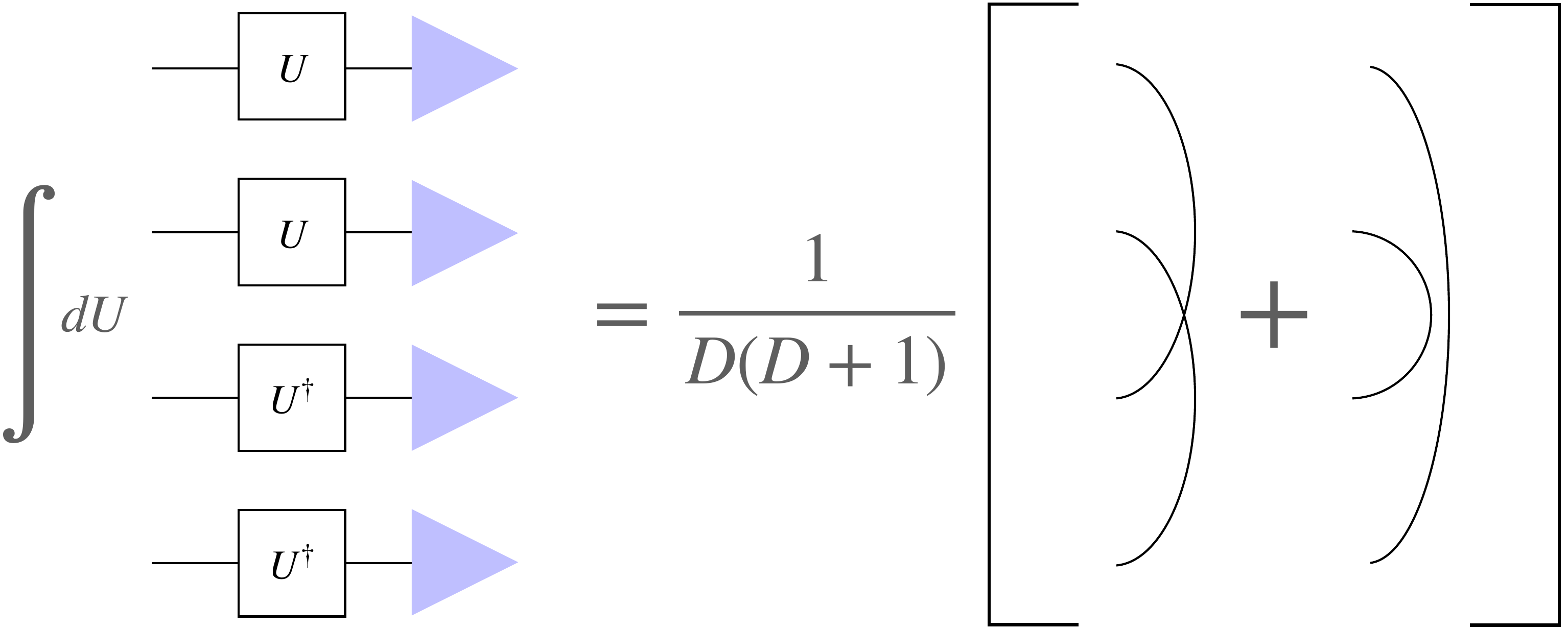}
    \end{gathered} \ .
\end{align}
Here the reference state vector is visualized as a blue triangle as in Fig.~\ref{fig:equilibration-intro}(a). In the calculations below, we will use a notation of Ref.\ \cite{Haferkamp_2021} to write \eqref{EQ_WEINGARTEN_STATE_4} as
\begin{align}
    \frac{1}{D(D+1)} \left[ \ket{1} + \ket{F} \right] \ ,
\end{align}
with the terms in the sum corresponding to the ``identity'' and ``flip'' state vectors in a quadrupled Hilbert space (also referred to as vectorized form \cite{Mele_2024}) of dimension $D^4$. The inner products between these states are given by
\begin{align}
\label{EQ_ID_FLIP_INNER}
    \braket{1}{1}&=\braket{F}{F}=D^2 \ , &
    \braket{1}{F}&=\braket{1}{F}=D \ ,
\end{align}
a result that directly follows from the graphical representation \eqref{EQ_WEINGARTEN_STATE_4}: Combining the same state produces two loops, while combining different ones produces only one.

\subsection{Random tensor networks}

We consider tensor networks on a graph $G$ with three types of edges (tensor legs): \emph{Internal} or contracted edges of (bond) dimension $b$, \emph{logical} edges of dimension $d$ that are associated with bulk degrees of freedom on each vertex/tensor, and \emph{physical} or open edges of dimension $a$ on the boundary of the tensor network. 
A contracted tensor network then furnishes a bulk-to-boundary map $V:\mathcal{H}_\text{bulk} \to \mathcal{H}_\text{bdy}$, where $\operatorname{dim} \mathcal{H}_\text{bulk} = d^{n_v}$ and $\operatorname{dim} \mathcal{H}_\text{bdy} = a^{n}$, $n_v$ and $n$ being the number of vertices and physical edges, respectively. The total number of contracted edges will be denoted by $n_\text{int}$.  
We will often consider tensor networks where the logical legs of the tensor are projected onto a bulk state vector $\ket{\Phi}$, with the resulting boundary state vector $\ket{{\psi}} = V \ket{\Phi}$ simply referred to as the \emph{tensor network state}. 
Though for most purposes we will assume that each type of edge is associated with the same dimension, we will in some cases consider the more general setting where each vertex with index $k$ can have a different logical dimension $d_k$, and each edge $\langle j,k\rangle$ between two indices $j$ and $k$ can have bond dimension $b_{\langle j,k\rangle}$. 
Following the notation of 
Ref.\ \cite{Haferkamp_2021}, we also define the total dimension of the $k$-th tensor as
\begin{equation}
\label{EQ_DIM_TENSOR}
    q_k \coloneqq a^{l_k-m_k}\, d_k \prod_{i=1}^{m_k} b_{k,i}\ ,
\end{equation}
in terms of the number $l_k$ of its planar (i.e., non-bulk) legs and $m_k$ planar internal legs, and where $b_{k,i}$ denotes the bond dimension of its $i$-th internal leg. For tensor networks where $a=b$ is constant for all tensors, this simplifies to $q_k = d_k\, b^{l_k}$.

Using these definitions, we are now in the position to formally define RTNs as follows.

\begin{dfn}[Random tensor networks]\label{DEF_RTN}
Let $G=(V,E)$ be a connected undirected finite graph of vertices $V$ and edges $E=E_\text{int} \cup E_\text{ext}$, the latter divided into \emph{internal} edges $\langle j, k \rangle$ and \emph{external} edges $\langle i, \emptyset\rangle$, $i,j,k \in V$, with external edges connected to only one vertex.
We map each graph to a tensor network by taking a tensor for each vertex, contracting over internal edges, and identifying the external edges with 
physical sites for the 
tensor network state vector $\ket{\psi} = \ket{\psi(G)}$.
Each tensor (before contraction) on vertex $i$ can be identified with a state vector 
\begin{equation}
    \ket{\psi_i} \in (\mathbb{C}^a)^{\otimes (l_i-m_i)} \otimes \mathbb{C}^d \bigotimes_{j \in \mathrm{adj}_i} \mathbb{C}^{b_{\langle i,j \rangle}} \ ,
\end{equation}
where $\mathrm{adj}_i$ denotes the vertices 
adjacent to the $i$-th vertex, and $a,b_{\langle i,j \rangle},d$ denote physical, internal, 
and logical dimensions as defined above.
We then define a \emph{tensor network state} vector as
\begin{equation}
     \ket{\psi}\coloneqq \bra{\Phi} \bigotimes_{\langle j, k \rangle \in E_\text{int}} \bra{j\;k} \bigotimes_{i \in V} \ket{\psi_i} \in (\mathbb{C}^a)^{\otimes n},
\end{equation}
where $\ket{\Phi}$ is a logical \emph{bulk state}  vector and $\ket{j\;k} \in \mathbb{C}^{b_{\langle j,k \rangle}}\otimes \mathbb{C}^{b_{\langle j,k \rangle}}$ is a maximally entangled EPR pair between the two endpoints of the edge $\langle j, k \rangle$.
We call $\ket{\psi}$ a \emph{random tensor network state} vector if each tensor is chosen as
\begin{equation}
    \ket{\psi_i} = U_{\text{Haar}} \ket{0} \ , 
\end{equation}
where $U_\text{Haar}$ is a Haar-random unitary chosen independently for each site $i$ acting on an arbitrary reference state vector $\ket{0}$. 
\end{dfn}

In general, the tensor network state vectors $\ket{\psi}$ will be unnormalized. Normalizing each sample when computing expectation values $\mathbb{E}[ {\bra{\psi}\mathcal{O}\ket{\psi}}/{\braket{\psi}{\psi}}]$ of the sample average for observables $\mathcal{O}$ is often cumbersome, while dividing the expectation value by the expectation value of the norm, i.e., computing ${\mathbb{E}[\bra{\psi}\mathcal{O}\ket{\psi}]}/{\mathbb{E}[\braket{\psi}{\psi}}]$, is much simpler.
Though this is a standard approach in the RTN literature \cite{Akers_2022, Hayden_2016}, one may worry that quantities such as \eqref{EQ_EFF_DIM_DEF} defined for the first averaging method may deviate from computations via the second one.
In App.~\ref{APP:FLUCTUATIONS}, we, therefore, derive a bound on the fluctuations of the norm $\braket{\psi}{\psi}$ for a very general class of RTN states and show that these vanish in the limit of large physical and bond dimensions, with fast convergence. We also show that this implies tight bounds on the expectations of properties of normalized random tensor network states. Numerical checks for systems such as those in Fig.~\ref{fig:equilibration-intro}(b) also confirm that this approach is unproblematic, and that the norm converges to its expectation value with far fewer samples than the observables we study here.

\section{Results}
\label{SEC_RESULTS}

\subsection{Equilibration of random tensor networks}

We now state the main technical result of our work, the inverse effective dimension for an RTN on a general graph $G$, which we derive from the Weingarten formulas of the previous section.
Rather than computing the effective dimension \eqref{EQ_EFF_DIM_DEF} for a full set of energy eigenstates of a specific Hamiltonian $H$, we will be assuming that almost all state vectors $\ket\psi$ in the RTN ensemble have a flat (small) overlap with each of the eigenstates, allowing us to write
\begin{equation}
    \frac{1}{D_\text{eff}} = a^n \frac{\mathbb{E} [ |\braket{\psi}{\phi}|^4 ]}{(\mathbb{E} [ \braket{\psi} ])^2} \ , 
\end{equation}
where $a^n$ is the dimension of the physical (boundary) Hilbert space and $\ket\phi$ is a reference state vector representing a typical eigenstate of $H$. This approach was previously used for random matrix product states (RMPS) in Ref.\ \cite{Haferkamp_2021}.

\begin{lem}[RTN effective dimension]
\label{LEM_EFF_DIM}
For a \emph{random tensor network} (RTN) state vector $\ket\psi$ with $n_\text{v}$ tensors, $n$ physical (external) boundary edges, and a logical bulk state vector $\ket{\Phi}$, the inverse effective dimension as defined by the on-average normalized fourth-power overlap with a reference state vector $\ket\phi$ acting on the external edges is upper-bounded by
\begin{align}
\label{EQ_EFF_DIM}
    &\frac{1}{D_\text{eff}} 
    \leq 
    \frac{Z_{\log\sqrt{b}}}{a^{n}} \prod_{k=1}^{n_\text{v}} \frac{1}{1+\frac{1}{q_k}} \ ,
\end{align}
with the bound saturated for non-entangled $\ket{\Phi}$ and $\ket{\phi}$.
Here $q_k$ is the total dimension \eqref{EQ_DIM_TENSOR} of the $k$-th tensor and  
$Z_{J}$ is defined as a rescaled classical Ising partition function at unit coupling,
\begin{align}
    Z_{J} &= \frac{1}{N}\sum_{\sigma}  e^{\sum_{\langle i,j \rangle \in E_\text{int}} J_{\langle i,j\rangle} \sigma_i \sigma_j} \ , \\
    N &= e^{\sum_{\langle i,j \rangle \in E_\text{int}} J_{\langle i,j\rangle}},
\end{align}
summing over $n_\text{v}$ classical spins taking values $\sigma_k \in \{-1,1\}^{}$ and is normalized by the equal-spin contribution $N$.
With this normalization, $Z_{J}$ effectively assigns factors $1$ and $e^{-2J_{\langle i,j\rangle}}=\frac{1}{b_{\langle i,j\rangle}}$ to aligned and anti-aligned spins, respectively, along each internal edge $\langle i,j\rangle \in E_\text{int}$ with bond dimension $b_{\langle i,j\rangle}$.
\end{lem}

\noindent
\textit{Proof.} We begin with the numerator of \eqref{EQ_EFF_DIM}. The expectation value of the norm $\braket{\psi}{\psi}$ of the RTN state vector $\ket\psi$ follows from the second-order Weingarten formula \eqref{EQ_WEINGARTEN_STATE_2}. To be instructive, we first compute this expression for a simple example of an RTN consisting of only two tensors. Following the style of Fig.\ \ref{fig:equilibration-intro}(a) but omitting labels except for the random unitaries, we write
\begin{align}
\label{EQ_TWO_TENSOR_EX}
    \ket{\psi} = 
    \begin{gathered}
        \includegraphics[width=0.4\linewidth]{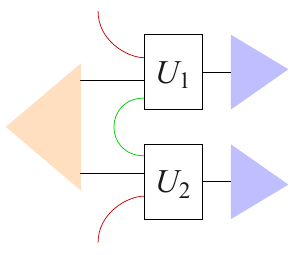}
    \end{gathered}
\end{align}
Here we have drawn physical bonds with dimension $a$ in red, internal bonds with dimension $b$ in green, and bonds connecting normalized states in black. 
Graphically, the integral yielding $\mathbb{E} [ \braket{\psi} ]$ can then be visualized as
\begin{align}
\label{EQ_EXP_NORM_EX}
    \begin{gathered}
        \includegraphics[width=0.7\linewidth]{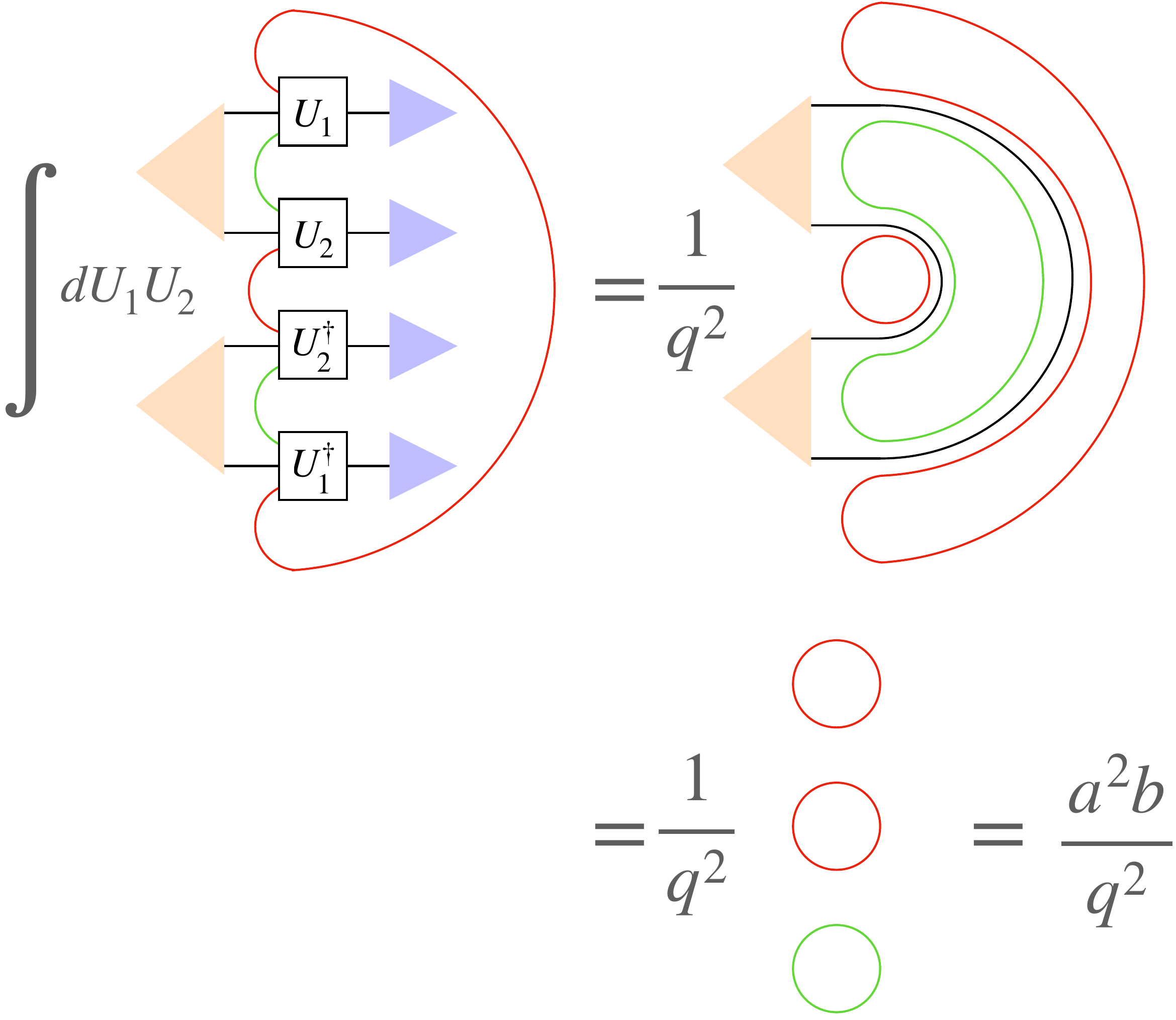}.
    \end{gathered}
\end{align}
The indices of adjoint unitaries are mirrored vertically for ease of visualization.
We denoted with $q$ the dimension of each Haar-random unitary, leading to a total dimension $q^2=(a\, b\, d)^2$ of the two tensors before contraction. In the second step, we found that the logical bulk state (orange triangle) drops out, as it only contributes a normalization factor $\braket{\Phi}{\Phi}=1$. We then end up with a loop (trace) for each internal (green) and external bond (red), contributing a factor of the respective dimension of the bond.
Generalized to an arbitrary RTN, this expression then becomes
\begin{align}
\label{EQ_G_NORM}
    \mathbb{E} [ \braket{\psi} ] = a^{n}  \prod_{\langle i,j \rangle \in E_\text{int}} b_{\langle i,j \rangle}\; \prod_{k=1}^{n_\text{v}} \frac{1}{q_k}  \ .
\end{align}
For simple tensor networks without logical legs ($d_k=1$), constant bond dimensions $a=b_{\langle i,j \rangle} = b$ and tensors with $l$ legs each, this expression simplifies to $b^{n+n_\text{int}-l\,n_\text{v}}=b^{-n_\text{int}}$.

Next we compute the denominator of \eqref{EQ_EFF_DIM}, given by the fourth power of the overlap between the RTN state vector $\ket\psi$ and a normalized reference state vector $\ket\phi$.
Applying the fourth-order Weingarten formula \eqref{EQ_WEINGARTEN_STATE_4} onto our example, we find for $\mathbb{E} [ |\braket{\psi}{\phi}|^4 ]$, the diagram
\begin{widetext}
\begin{align}
    \label{fullcalc}
    \raisebox{-0.5\height}{\includegraphics[width=0.14\linewidth]{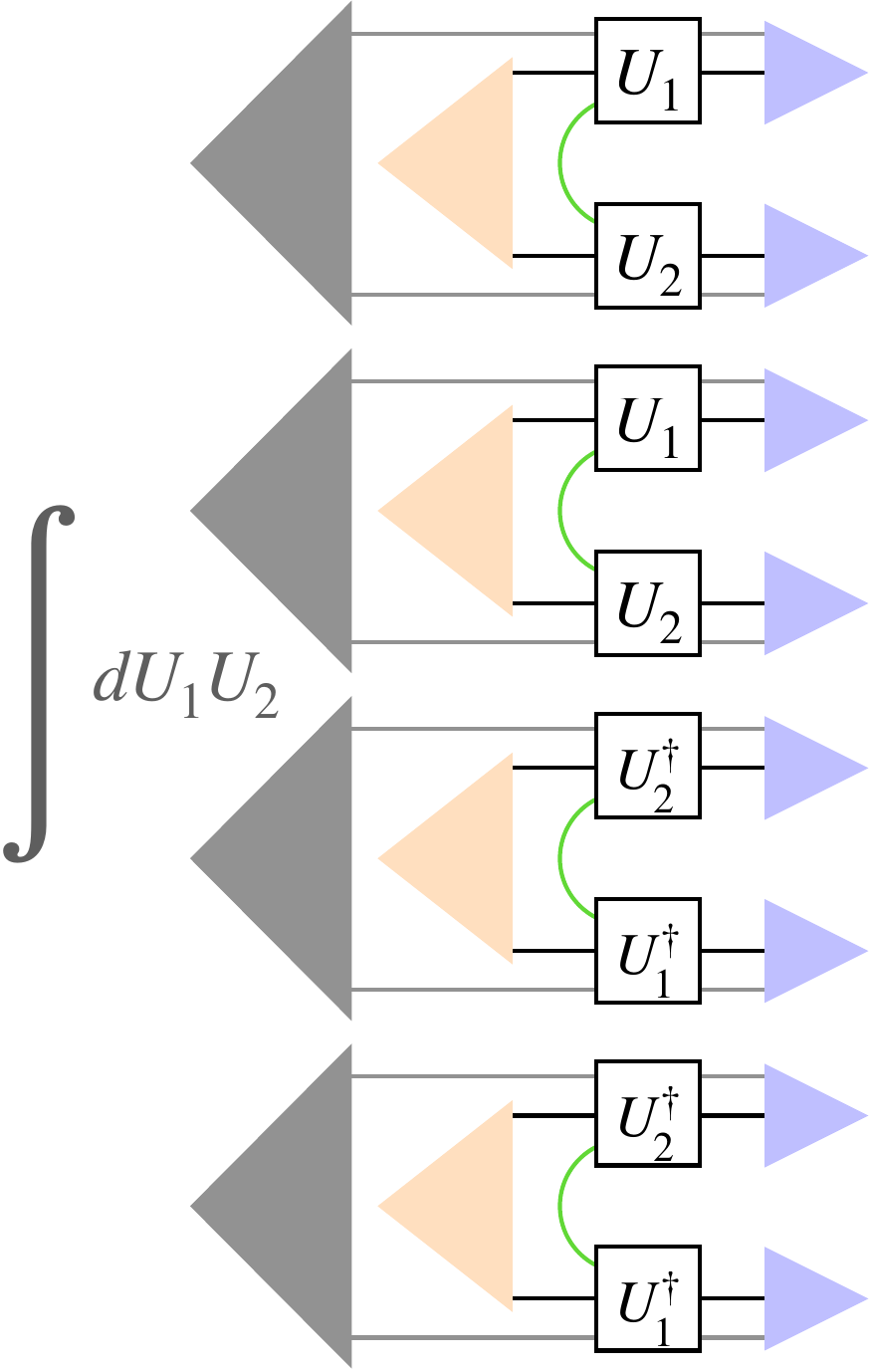}}
    &\mathrel{\vcenter{\hbox{=}}}
    \raisebox{-0.5\height}{\includegraphics[width=0.68\linewidth]{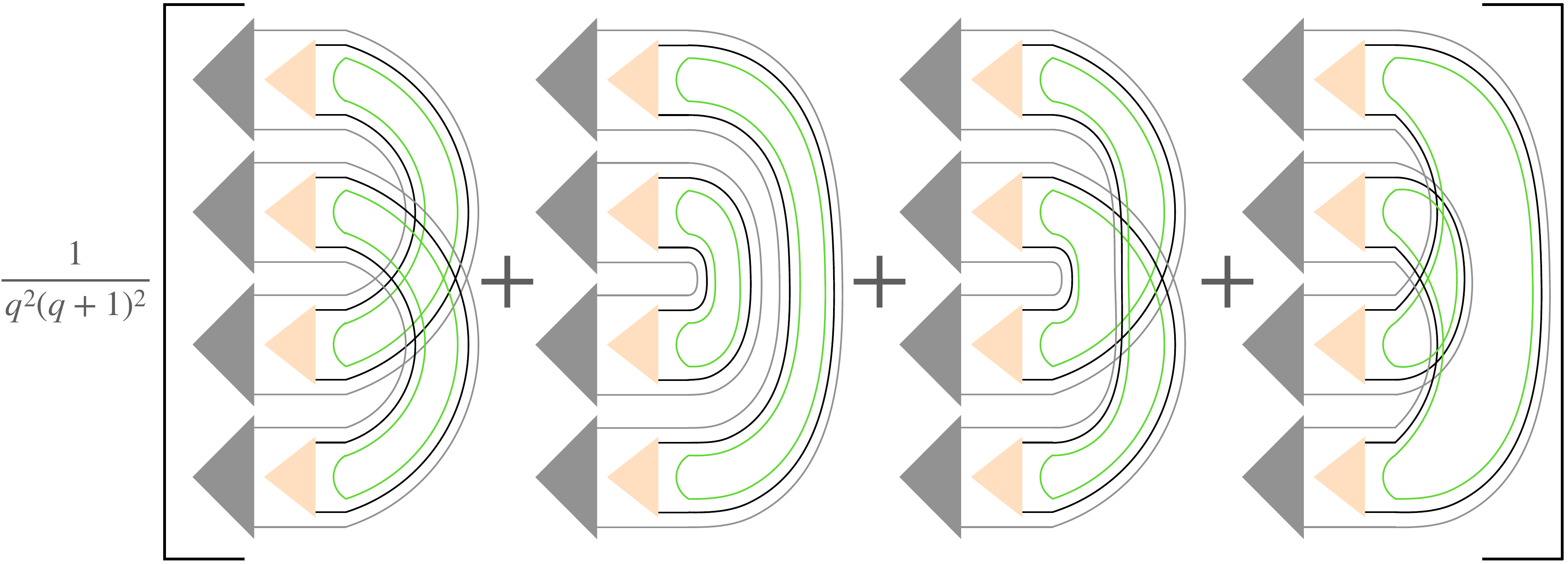}} \nonumber \\
    &\mathrel{\vcenter{\hbox{=}}}
    \raisebox{-0.5\height}
    {\includegraphics[width=0.55\linewidth]{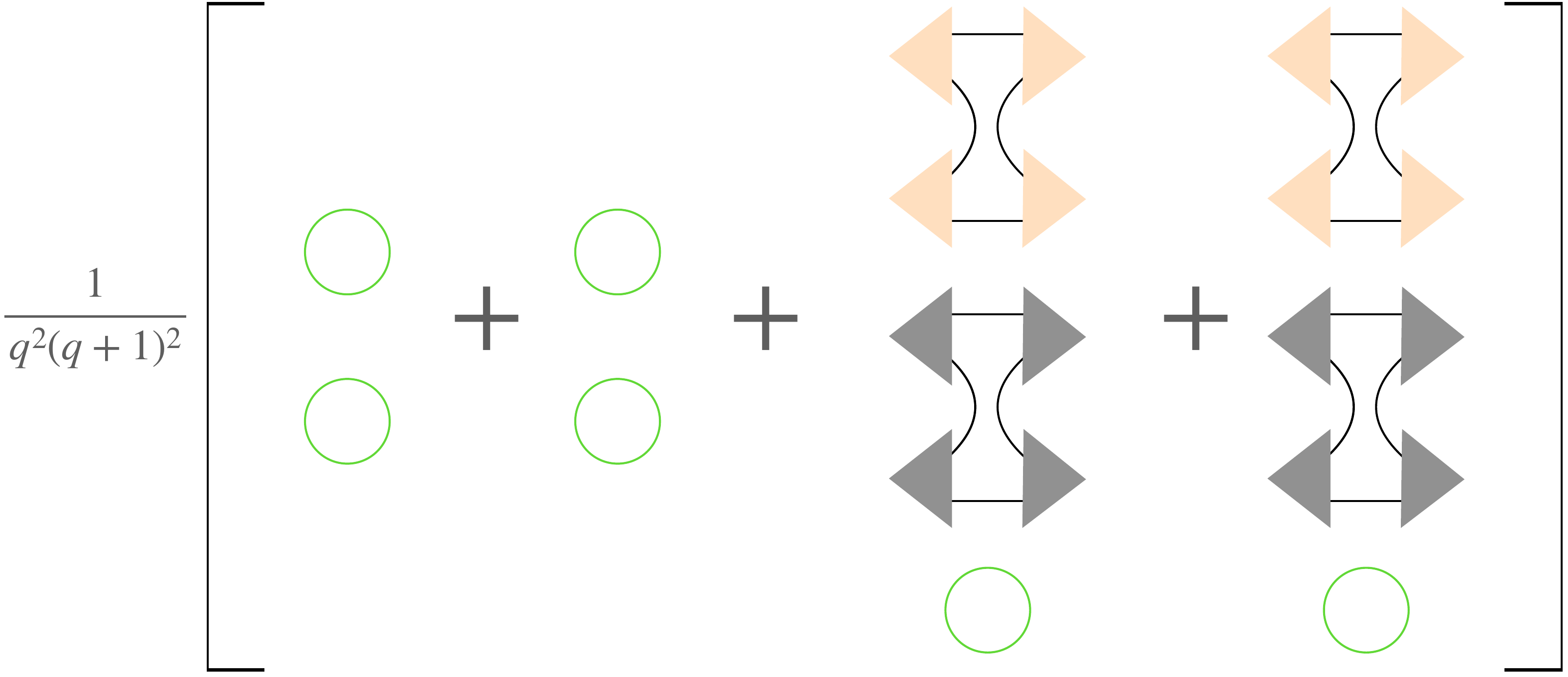}}\\
    &\mathrel{\vcenter{\hbox{$=$}}}
\frac{2(b^2+b\tr_1[\rho_1^2]\tr_1[\varsigma_1^2])}{q^2(q+1)^2} \nonumber \mathrel{\vcenter{\hbox{$\leq$}}}
    \frac{2(b^2+b)}{q^2(q+1)^2} .\nonumber
\end{align}
\end{widetext}
Here we have added a factor $\frac{1}{q(q+1)}$ for each tensor and then found a sum of four terms that, using the identity/flip state formalism, can be written as
\begin{align}
\label{EQ_W4_EX_NUM}
    &\bra{\rho}^{\otimes 2}\bra{\varsigma}^{\otimes 2}\bra{\chi_{1,2}}^{\otimes 2} \left( \ket{1}\ket{1} + \ket{F}\ket{1} + \ket{1}\ket{F} + \ket{F}\ket{F}\right) \nonumber\\
    &=  \braket{1}{1} + \braket{F}{F} + \left( \braket{1}{F} + \braket{F}{I} \right) \tr_1[\rho_1^2] \tr_1[\varsigma_1^2] 
    \nonumber\\
    &\leq 2(b^2 + b)\ ,
\end{align}
where we have transformed the bulk and boundary density matrices $\rho=\ketbra{\Phi}{\Phi}$ and $\varsigma = \ketbra{\phi}{\phi}$ into a state-vector representation $\ket{\rho} = \ket{\Phi}\ket{\Phi^\dagger}$ and $\ket{\varsigma} = \ket{\phi}\ket{\phi^\dagger}$, and $\bra{\chi_{1,2}}$ represents a reduced density matrix $\chi_{1,2}=\ketbra{1\, 2}{1\, 2}$ of an EPR state vector $\ket{1\, 2}$ effecting the contraction between tensor 1 and 2.
In the second step, we have introduced the reduced density matrices as traces of $\rho$ and $\varsigma$ over the second site,  
\begin{align}
    \rho_1 &= \tr_2[\rho] \ , &
    \varsigma_1 &= \tr_2[\varsigma] \ .
\end{align}
These can alternatively be written in terms of the second R\'enyi entropy, e.g.,\ $\tr_1[\rho_1^2] = e^{-2S_2(\rho_1)}$. The R\'enyi entropies are non-negative, which immediately implies that these traces are upper-bounded by $1$, a fact that is used in the last step of \eqref{EQ_W4_EX_NUM}.
This inequality is saturated if both the bulk state vector $\ket{\Phi}$ and boundary state vector $\ket{\phi}$ are products, i.e., no entanglement is present between different sites.
We have committed a slight abuse of notation from the first to second step of \eqref{EQ_W4_EX_NUM}, as in the first each $\ket{1}$ and $\ket{F}$ act on a Hilbert space of dimension $q^4$, whereas in the second, the dimension is $b^4$ (with the $\ket{1,F}$ only acting on internal edges of bond dimension $b$). The inner products were then resolved via \eqref{EQ_ID_FLIP_INNER}.

We now move beyond our example to the general case of this fourth-power overlap, again utilizing the identity/flip formalism.
Avoiding any further abuse of notation, we define for each vertex $k$, with $m_k$ internal and $l_k-m_k$ external legs, the states
\begin{equation}
    \ket{1,F}_k \coloneqq \ket{1,F}_\text{int}^{\otimes m_k} \otimes \ket{1,F}_\text{phys}^{\otimes (l_k-m_k)} \otimes \ket{1,F}_\text{log} \ ,
\end{equation}
written as tensor products over internal, physical, and logical bonds. Contractions will only lead to inner products between internal bonds, which we can resolve via
\begin{align}
    \braket{1}{1}_\text{int} &=\braket{F}{F}_\text{int}= b^2 \ , & \braket{1}{F}_\text{int} &= b \ ,
\end{align}
where $b=b_{\langle j,k \rangle}$ is the bond dimension of each internal bond, possible varying with each edge $\langle j,k \rangle$ adjacent to the $k$-th vertex.
With this definition, we can then write the general 
fourth-power overlap as
\begin{align}
\label{EQ_G_OVERLAP4}
    \mathbb{E} [ |\braket{\psi}{\phi}|^4 ] = \bra{\rho}^{\otimes 2}\bra{\varsigma}^{\otimes 2} \prod_{\langle i,j \rangle \in E_\text{int}} \bra{\chi_{i,j}}^{\otimes 2} \, \prod_{k=1}^{n_\text{v}} \frac{\ket{1}_k + \ket{F}_k}{q_k(q_k+1)} \ .
\end{align}
The first product in this expression contains the EPR pairs with which each tensor contraction over internal edges is performed, and the second a term \eqref{EQ_WEINGARTEN_STATE_4} for each tensor.
Resolving each inner product between identity and flip states via \eqref{EQ_ID_FLIP_INNER} is then equivalent to summing terms in an Ising model on a graph $G$ made from the vertices (tensors) and internal edges (contractions): Replacing $\ket{1}$ and $\ket{F}$ on each vertex with classical spins $\sigma_k = \pm 1$, we sum over all spin configurations and take a product of partition function factors over all edges $\langle i,j\rangle$ ($b^2$ for aligned spins and $b$ for anti-aligned ones), which can equivalently be written as an action for a classical Hamiltonian.
Including the expected normalization \eqref{EQ_G_NORM}, we write
\begin{align}
    \frac{\mathbb{E} [ |\braket{\psi}{\phi}|^4 ]}{(\mathbb{E} [ \braket{\psi} ])^2} 
    = \frac{\bra{\rho}^{\otimes 2}\bra{\varsigma}^{\otimes 2}}{a^{2 n}} &\prod_{\langle i,j \rangle \in E_\text{int}} \frac{\bra{\chi_{i,j}}^{\otimes 2}}{b_{\langle i,j \rangle}^2} \, \prod_{k=1}^{n_\text{v}} \frac{\ket{1} + \ket{F}}{1 +  \frac{1}{q_k}} \nonumber\\
    =\frac{1}{a^{2 n}}\prod_{k=1}^{n_\text{v}}   \frac{1}{1 +  \frac{1}{q_k}} 
    &\sum_{\sigma_k \in \{-1,1\}^{n_\text{v}}} \tr_a[\rho_a^2] \tr_A[\varsigma_A^2]  \nonumber\\
    \cdot&\prod_{\langle i,j \rangle \in E_\text{int}} 
    \begin{cases}
        1 & \text{if } \sigma_i=\sigma_j, \\
        \frac{1}{b} & \text{if } \sigma_i \neq \sigma_j,
    \end{cases}
\end{align}
where we have defined $\rho_a$ and $\varsigma_A$ as the reduced density matrices over the bulk and boundary-adjacent subregions $a=a(\sigma)$ and $A=A(\sigma)$ with spin-up configurations (or by symmetry of the pure-state second Renyi entropy, spin-down). These are again strictly-upper bounded by $1$ (when $\ket{\Phi}$ and $\ket{\phi}$ are product states), leading us the bound
\begin{align}
    \frac{\mathbb{E} [ |\braket{\psi}{\phi}|^4 ]}{(\mathbb{E} [ \braket{\psi} ])^2}  \leq\frac{1}{a^{2 n}}\prod_{k=1}^{n_\text{v}} \frac{1}{1 +  \frac{1}{q_k}}   \underbrace{\sum_{\sigma_k \in \{-1,1\}^{n_\text{v}}} 
    \prod_{\langle i,j \rangle \in E_\text{int}} b^{\frac{\sigma_i \sigma_j-1}{2}} }_{=Z_{\log\sqrt{b}}} \ ,
\end{align}
which, including the additional $a^n$ factor, is equal to \eqref{EQ_EFF_DIM}. \qed

In order to turn the upper bound computed here into an equality, one would have to extend the Ising Hamiltonian with a term proportional to the second Renyi entropies $S_2(\rho_A)$ and $S_2(\varsigma_A)$, 
which is generally not solvable for arbitrary bulk and boundary states.
However, for sufficiently small entanglement in either state we can assume the bound to be somewhat tight.

Lemma \ref{LEM_EFF_DIM} immediately implies that an RTN on a fixed graph $G$ equilibrates for large boundary dimension:

\begin{cor}[RTN equilibration limit]
\label{COR_RTN_EQUI_LIMIT}
    In the limit of large physical dimension $a$, an RTN exhibits perfect equilibration, i.e., 
    \begin{equation}
        \lim_{a \to \infty} \frac{1}{D_\text{eff}} = 0 \ ,        
    \end{equation}
    for any (finite or infinite) choice of bond and logical dimensions $b,d \geq 1$.
\end{cor}

\noindent
\textit{Proof.} 
We first note that the partition function $Z_{\log\sqrt{b}}$ only contains zero or negative powers of the internal bond dimension $b$, and can therefore be always upper-bounded by its value for $b=1$, given by $Z_{0} = 2^{n_v}$.
It then follows from \eqref{EQ_EFF_DIM} that
\begin{align}
    \frac{1}{D_\text{eff}} \leq \frac{2^{n_v}}{a^{n}} \ ,
\end{align}
and hence for fixed $n_v$ and $n \geq 1$, $\lim_{a \to \infty} \frac{1}{D_\text{eff}} = 0$. \qed

We can also consider the limit of \eqref{EQ_EFF_DIM} at large bond dimension $b$, where the ``zero-temperature limit'' $\lim_{b\to\infty}Z_{\log\sqrt{b}}=2$ of the Ising partition function implies that 
\begin{equation}
\label{EQ_RTN_EFF_DIM_INF_B}
    \lim_{b\to\infty} \frac{1}{D_\text{eff}} = \frac{2}{a^{n}} \ ,
\end{equation}
if we assume that each vertex $k$ is connected to at least one internal bond, leading to $q_k\to\infty$ and a simplification of the product term in \eqref{EQ_EFF_DIM}. 
Similarly, if the logical (bulk) dimension $d$ on each tensor becomes large, this product terms becomes unity and we find
\begin{equation}
\label{EQ_RTN_EFF_DIM_INF_D}
    \lim_{d\to\infty} \frac{1}{D_\text{eff}} = \frac{Z_{\log\sqrt{b}}}{a^{n}} \ .
\end{equation}
In fact, we find that increasing the logical dimension produces weaker equilibration, up the bound \eqref{EQ_RTN_EFF_DIM_INF_D}:
\begin{cor}[RTN equilibration dependence on logical dimension]
\label{COR_RTN_LOG_DIM}
    The inverse effective dimension of an RTN for unentangled bulk and boundary reference state vectors $\ket{\Phi}$ and $\ket{\phi}$ increases monotonically with the logical dimension $d_k$ on any vertex $k$.
\end{cor}

\noindent
\textit{Proof.} 
We assume the initial RTN has inverse effective dimension $D_\text{eff}$ and logical dimension $d_k$ on the $k$-th vertex. We then construct a new RTN with inverse effective dimension $D_\text{eff}^\prime$ and logical dimension $d_k^\prime>d_k$ on that vertex. From \eqref{EQ_EFF_DIM}, which becomes an equality for product state vectors $\ket{\Phi}$ and $\ket{\phi}$, we then find
\begin{align}
    \frac{D_\text{eff}}{D_\text{eff}^\prime} = \frac{q_k+1}{q_k+\frac{d_k}{d_k^\prime}} > 1 \ ,
\end{align}
and hence $1/D_\text{eff}^\prime > 1/D_\text{eff}$. \qed

\subsection{Equilibration under RTN deformations}

We now consider how modifications of the RTN geometry change the effective dimension. As the elementary deformation step, we consider \emph{vertex fusion}, where we take two 
vertices (random tensors) $j$ and $k$ connected by an edge $\langle j,k \rangle$ and replace both by a single vertex connected to 
all vertices that $j$ and $k$ have  previously been connected to (excluding 
$\langle j,k \rangle$)
\begin{equation}
    \begin{gathered}
    \includegraphics[width=0.35\textwidth]{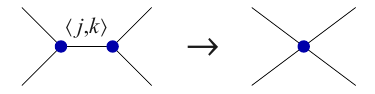}
    \end{gathered}.
\end{equation}
The new tensor is then generated by 
a Haar-random unitary of dimension 
\begin{equation}
q=\frac{q_j q_k}{b_{\langle j,k \rangle}},
\end{equation}producing more non-local randomness than the previous two unitaries of dimension $q_j$ and $q_k$. We can formalize the effect on equilibration as follows. 

\begin{lem}[Effective dimension growth under vertex fusion]
\label{LEM_RTN_EFF_DIM_GROWTH}
    Consider an RTN with $n_v \geq 2$ vertices, $n_\text{e,int} \geq 1$ internal edges, and effective dimension $D_\text{eff}$ for unentangled bulk and boundary reference state vectors $\ket{\Phi}$ and $\ket{\phi}$. Fusing any pair of connected vertices into a single vertex yields a new RTN with $n_v-1$ vertices and $n_\text{e,int}-1$ internal edges with effective dimension $D_\text{eff}^\prime$ satisfying
    \begin{equation}
        D_\text{eff} \to D_\text{eff}^\prime > D_\text{eff} \ .
    \end{equation}
\end{lem}
\noindent
\textit{Proof.} Assume we are fusing the connected vertices $j$ and $k$. From \eqref{EQ_EFF_DIM}, the ratio between the effective dimensions of both RTNs is given by
\begin{align}
\label{EQ_EFF_DIM_FRAC}
    \frac{D_\text{eff}}{D_\text{eff}^\prime} &= \frac{Z^\prime\, ( 1 + \frac{1}{q_j} ) ( 1 + \frac{1}{q_k} )}{Z\, ( 1 + \frac{1}{q_{j \circ k}} )}
\end{align}
where we write $Z \coloneqq Z_{\log\sqrt{b}}$ and $Z^\prime \coloneqq  Z_{\log\sqrt{b}}^\prime$ as a shorthand for the original and fused Ising partition functions, and $q_{j \circ k} \coloneqq  \frac{q_j q_k}{b_{\langle j,k \rangle}}$ for the dimension of the final vertex, $b_{\langle j,k \rangle}$ being the bond dimension of the edge connecting the two initially fused vertices.
Taking a closer look at the partition functions of each graph, we find that $Z^\prime$ contains exactly those terms of $Z$ with aligned classical spins $\sigma_j = \sigma_k$ (with unit energy along the edge $\langle j,k \rangle$), while $\Delta Z \coloneqq Z - Z^\prime \leq Z^\prime$ contains those with anti-aligned spins (suppressed by a $1/b_{\langle j,k \rangle}$ factor).
Using \eqref{EQ_Z_RELATIONS_GENERAL} as derived in App.~\ref{APP:PROOF1}, we then bound
\begin{align}
     \frac{D_\text{eff}}{D_\text{eff}^\prime} &= \frac{( 1 + q_j ) ( 1 + q_k )}{\left(1 + \frac{\Delta Z}{Z^\prime} \right) \left( b_{\langle j,k \rangle}^2 + q_j q_k \right)}  \nonumber\\
     &\leq \frac{( 1 + q_j ) ( 1 + q_k )}{\left(1 + \frac{q_j + q_k}{b_{\langle j,k \rangle}^2 + q_j\, q_k} \right) \left( b_{\langle j,k \rangle}^2 + q_j q_k  \right)} \nonumber\\
     &= \frac{( 1 + q_j ) ( 1 + q_k )}{ \underbrace{b_{\langle j,k \rangle}^2 - 1}_{\geq 0} + ( 1 + q_j ) ( 1 + q_k )}   
     \leq 1 \ . 
\end{align}
Hence $D_\text{eff} \leq D_\text{eff}^\prime$, where equality can only be reached for $b_{\langle j,k \rangle}=1$, i.e., when fusing two vertices that were previously unconnected by an edge.
\qed

We can now derive a bound from this result by considering the case where all vertices of a any given RTN are successively fused together, until only a single random tensor of high dimension -- typically interpreted as a ``black hole'' in holographic tensor network models -- remains: 
\begin{cor}[Bound on effective dimension]
\label{COR_EFF_DIM_ABS_LOWER_BOUND}
    The inverse effective dimension of any RTN state on $n$ physical sites for unentangled bulk and boundary reference state vectors $\ket{\Phi}$ and $\ket{\phi}$ can be lower-bounded by 
    \begin{equation}
    \label{EQ_EFF_DIM_MIN}
         \frac{1}{D_\text{eff}} \geq \frac{2}{a^{n} + 1} \ ,
    \end{equation}
    which is achieved by a single random tensor with $n$ legs.
\end{cor}

\noindent
\textit{Proof.} 
By Lemma \ref{LEM_RTN_EFF_DIM_GROWTH}, the inverse effective dimension of an RTN can be iteratively reduced by fusing vertices until only a single tensor with $n$ legs remains.
From \eqref{EQ_EFF_DIM}, its effective dimension is given by 
\begin{align}
    \frac{1}{D_\text{eff}} = \frac{2}{a^{n}+\frac{1}{d}} \ .
\end{align}
It becomes minimal for $d=1$ (no bulk leg), where it reaches the bound \eqref{EQ_EFF_DIM_MIN}. \qed
\\
Note that this implies the simple bound $D_\text{eff} < \frac{1}{2} \dim \mathcal{H}_\text{phys}$, where $\mathcal{H}_\text{phys} = a^n$ is the physical Hilbert space dimensions, for any RTN.

We can now make an interesting comparison: For a large physical dimension $a$ or a large number of sites $n$, the inverse effective dimension of a single random tensor approaches the limit \eqref{EQ_RTN_EFF_DIM_INF_B} of large bond dimension $b$ of \emph{any} RTN. In other words, a sufficiently large RTN at finite physical dimension $a$ can be approximated by a single random tensor in the limit of infinite bond dimension $b$; in that limit, each tensor produces so much randomness so that the RTN geometry ceases to matter. An intuitive explanation for this is that as bulk locality is lost, the boundary state vector becomes more delocalized. 

Note that other graph operations, such as deleting edges (without fusing the vertices they connect) and deleting vertices (along with their connected edges) may increase or decrease $D_\text{eff}$ and thus no strict bounds apply.

We will now define upper and lower bounds for the Ising partition function using a recursive approach to evaluate its graph.
\begin{lem}[Upper and lower bounds for the Ising partition function]
    The partition function $Z(G) \equiv
    Z_{\log\sqrt{\beta}}(G)$ on the graph $G=(V,E_{\text{int}})$ of vertices $V$ and internal edges $E_\text{int}$ can be both upper- and lower-bounded via the partition function $Z(G \setminus k)$, where $G\setminus k$ denotes the graph $G$ with the $k$-th vertex and its adjacent edges are removed, as
\begin{align}
\label{EQ_Z_BOUND_LOWER}
   Z(G) &\geq \left( \frac{1}{b^{\lceil m_k/2 \rceil}} + \frac{1}{b^{\lfloor m_k/2 \rfloor}} \right)\, Z(G\setminus k) \ ,  \\
\label{EQ_Z_BOUND_UPPER}
   Z(G) &\leq \left(1 + \frac{1}{b^{m_k}} \right)\, Z(G\setminus k) \ ,
\end{align}
    where $b$ is the dimension of internal bonds (assumed constant) and $m_k$ denotes the number of internal edges adjacent to the $k$-th vertex, i.e., those connecting the vertex to $G\setminus k$.
    For $m_k=0,1$ these bounds coincide to form the exact equality $Z(G)=(1+\frac{1}{b^{m_k}})Z(G \setminus k)$.
\end{lem}

\noindent
\textit{Proof.} 
We first notice that the terms in $Z(G\setminus k)$ have a $\mathbb{Z}_2$ symmetry, as flipping all spins $\sigma_i \to -\sigma_i$ (for $i\neq k$) leaves the factors for each edge ($1$ or $1/b$ for aligned/anti-aligned spins) invariant. 
We can thus write 
\begin{equation}
    Z(G\setminus k)=\sum_{\alpha=1}^{2^{m_k-1}} Z_\alpha (G\setminus k) \ ,
\end{equation}
where each term in the sum includes contributions where the $m_k$ vertices adjacent to the $k$-th vertex (assuming no doubling of edges) are in one of the $2^{m_k}$ possible spin states or its spin-flip dual.
The full partition function is then given by
\begin{equation}
\label{EQ_Z_DECOMP_K}
    Z(G) = \sum_{\alpha=1}^{2^{m_k-1}} f_{\alpha,k} \, Z_\alpha (G\setminus k) \ ,
\end{equation}
where $f_{\alpha,k}$ includes additional factors arising from edge weights for anti-aligned spins in each of the two terms. Assuming $b>1$, we can thus bound
\begin{equation}
    \frac{1}{b^{\lceil m_k/2 \rceil}} + \frac{1}{b^{\lfloor m_k/2 \rfloor}} \leq f_{\alpha,k} \leq 1 + \frac{1}{b^{m_k}} \ ,
\end{equation}
which when inserted into \eqref{EQ_Z_DECOMP_K} yields the lower and upper bounds \eqref{EQ_Z_BOUND_LOWER} and \eqref{EQ_Z_BOUND_UPPER}.
 \qed
\\

As a simple example, consider a $L \times L$ square lattice: Removing vertices row-by-row, we find that for each of the first $L-1$ rows involved, we need to remove $L-1$ vertices with $m=2$ legs and one vertex with $m=1$ legs (note that removing the $k$-th vertex also involves all edges still connected to it). In the final row, we remove $L-1$ vertices with $m=1$ legs and one with $m=0$. This leads to the bounds
\begin{align}
\label{EQ_Z_44_BOUND_LOWER}
    Z^{L \times L} &\geq 2 \left( 1 + \frac{1}{b} \right)^{2(L-1)} \left( \frac{2}{b} \right)^{(L-1)^2} \ ,\\
\label{EQ_Z_44_BOUND_UPPER}
    Z^{L \times L} &\leq 2 \left( 1 + \frac{1}{b} \right)^{2(L-1)} \left( 1 + \frac{1}{b^2} \right)^{(L-1)^2} \ .
\end{align}
Together with the absolute lower bound \eqref{EQ_EFF_DIM_MIN} for general geometries, we will now be able to derive that specific tensor network geometries induce a strictly separate hierarchy of equilibration in terms of their effective dimensions.

\subsection{RTN equilibration hierarchy}

We now consider the form of the inverse effective dimension \eqref{EQ_EFF_DIM} for specific RTN geometries with $1$-dimensional boundaries.
Consider a \emph{matrix product state} (MPS), also called \emph{tensor-train state} (TTS), a one-dimensional chain of $n=n_\text{ext}$ tensors with one physical leg per tensor. An MPS can be defined with open or closed/periodic boundary conditions (the latter shown in Fig.\ \ref{fig:ising-tn-geo_eff-dim}(a)), leading to different partition functions.
To avoid confusion with the definition of a \emph{random matrix product state} (RMPS) in Ref.\ \cite{Haferkamp_2021}, where the MPS tensors are chosen as Haar-random unitaries $U:\mathcal{H}_{b} \otimes \mathcal{H}_{d} \to \mathcal{H}_{a} \otimes \mathcal{H}_{b}$ from 
one bond leg and the bulk leg to the boundary leg and the second bond leg, we use a different name for our RTNs on an MPS/TTS geometry.
\begin{dfn}[Random tensor train]
    A \emph{random tensor train} (RTT) is defined as a random tensor network (following Def.~\ref{DEF_RTN}) on an open or closed chain of tensors, each with Hilbert space dimension $a b^2 d$ or $a b d$ (for endpoints of an open chain), i.e., with one physical leg, one bulk leg, and one or two internal legs. 
\end{dfn}
\noindent
Note that at large bond dimension, random tensors approximate \emph{perfect tensors} and thus become approximately unitary for any bipartition into equal-dimensional sets. In this limit, the RTT and RMPS definitions coincide.
For the open case of an RTT with $n_\text{e,int} = n-1$ edges, each possible edge configuration (aligned or anti-aligned) can be fulfilled by two classical spin states, leading to
\begin{align}
    Z_{\log \sqrt{b}}^\text{RTT, open} = \sum_{k=0}^{n-1} \begin{pmatrix}
        n-1 \\
        k
    \end{pmatrix} \frac{2}{b^k}
    = 2 \left( 1 + \frac{1}{b}\right)^{n-1}
    \ ,
\end{align}
where $b$ is the dimension of internal bonds. Periodic boundary conditions imply $n_\text{e,int} = n$ edges but with only an even number of anti-aligned edges allowed. This leads to
\begin{align}
    Z_{\log \sqrt{b}}^\text{RTT, closed} = \sum_{k=0}^{\lfloor n/2 \rfloor} \begin{pmatrix}
        n \\
        2k
    \end{pmatrix}
    \frac{2}{b^{2k}}
    = \left( 1+\frac{1}{b} \right)^n + \left( 1 - \frac{1}{b} \right)^n \ .
\end{align}
Inserting both result into \eqref{EQ_EFF_DIM} then yields
\begin{align}
\label{EQ_EFF_DIM_MPS_OPEN}
    \frac{1}{D_\text{eff}^\text{RTT, open}} &= \frac{2 \left( 1 + \frac{1}{b}\right)^{n-1}}{a^{n}\,  \left( 1 + \frac{1}{a\, d\, b} \right)^2\left( 1 + \frac{1}{a\, d\, b^2} \right)^{n-2} 
    } \  , \\
\label{EQ_EFF_DIM_MPS_CLOSED}
    \frac{1}{D_\text{eff}^\text{RTT, closed}} &= \frac{\left( 1 + \frac{1}{b}\right)^{n}+\left( 1 - \frac{1}{b}\right)^{n}}{a^{n}\, \left( 1 + \frac{1}{a\, d\, b^2} \right)^n
    } \  .
\end{align}
Here $a$ and $d$ are the dimension of physical (boundary) sites and bulk sites, respectively. Most MPS constructions do not include bulk degrees of freedom, and hence set $d=1$. Again, these equalities turn into upper bounds on $1/D_\text{eff}$ for entangled bulk and reference states.
We show an example for $n=5$ sites, $b=a=2$ (spins), $d=1$ (no bulk sites), and closed boundary conditions in Fig.\ \ref{fig:equilibration-intro}(b). There 
we compare $\Delta A_\psi^\infty=1/\sqrt{D_\text{eff}^\text{RTT, closed}} \approx 0.188$ with the 
result for a single RTN sample time-evolved under an Ising Hamiltonian $H_I=\sum_k(X_k X_{{k+1}}+Z_k)$ and a normalized observable $A=X_1$, leading to fluctuations in $\langle A \rangle$ with standard deviation $\sigma_{\langle A \rangle} \approx 0.175$, 
close to the analytical upper bound $\Delta A_\psi^\infty$. Note that the bound is within $10\%$ of being saturated even though the energy eigenstates of $H_I$ are not product states.

Above, we showed in Cor.\ \ref{COR_RTN_EQUI_LIMIT} that every RTN with fixed geometry equilibrates in a \emph{continuum limit} of large physical dimension $a$. With the definition of an RTT above, we are now in a position to show asymptotic equilibration of a class of tensor networks in a \emph{scaling limit} where both the number of boundary sites $n$ and the number of tensors $n_v$ diverges:

\begin{lem}[RTT equilibration]
    In the scaling limit of infinitely many sites $n$, an RTT state equilibrates perfectly, i.e.,
    \begin{equation}
        \lim_{n \to \infty} \frac{1}{D_\text{eff}^\text{RTT}} = 0 \ ,
    \end{equation}
    for both open and closed boundary conditions and any (finite or infinite) choice of bond and logical dimensions $b,d \geq 1$ and a non-trivial physical 
    dimension $a \geq 2$.
\end{lem}

\noindent
\textit{Proof. }
The numerator of both \eqref{EQ_EFF_DIM_MPS_OPEN} and \eqref{EQ_EFF_DIM_MPS_CLOSED} can be upper-bounded by $2^n$ (its value for $b=1$), so that for any values of $a,b,d$ we can bound the inverse effective dimension as
\begin{equation}
    \frac{1}{D_\text{eff}^\text{RTT}} \leq \frac{2^n}{a^{n}} \ .
\end{equation}
It follows that $\lim_{n \to \infty} {1}/{D_\text{eff}^\text{RTT}} = 0$ for any $a>1$, which is fulfilled by our assumption $a \geq 2$. \qed

Next, we seek to establish the inverse effective dimension of an RTT as the upper bound of another class of RTNs, those on regular hyperbolic tilings.

\begin{dfn}
    [Regular tilings] A regular two-dimensional tiling, characterized by the Schl\"afli symbol $\{p,q\}$, is a uniform tiling with regular $p$-gons, $q$ of which meet at each vertex. The five configurations of integer $p$ and $q$ for which $\frac{1}{p}+\frac{1}{q} > \frac{1}{2}$ correspond to the surfaces of the five platonic solids (regular polyhedra), while the three cases fulfilling $= \frac{1}{2}$ and the infinitely many cases fulfilling $<\frac{1}{2}$ are known as \emph{flat} and \emph{hyperbolic} planar tilings, respectively.
\end{dfn}

An example of a hyperbolic $\{5,4\}$ tiling (or rather, a small subregion thereof) is given in Fig.~\ref{fig:ising-tn-geo_eff-dim}(c), and a hyperbolic $\{4,5\}$ tiling is shown in Fig~\ref{fig:rt-cuts}(c).

\begin{lem}[Hyperbolic vs.\ RTT equilibration]
\label{LEM_EFF_DIM_HYP_VS_RTT}
    The boundary state of an RTN on a regular $\{p,q\}$ hyperbolic tiling with $p>3$, and a sufficiently large number $n$ of physical boundary sites has an inverse effective dimension that is smaller than that of an RTT with the same boundary and internal bond dimensions, assuming unentangled bulk and boundary reference state vectors $\ket{\Phi}$ and $\ket{\phi}$ in both cases.
\end{lem}

\noindent
\textit{Proof. }
First consider the scaling of the inverse effective dimension of an RTT with $n$. Both for open and closed boundary conditions (\eqref{EQ_EFF_DIM_MPS_OPEN} and \eqref{EQ_EFF_DIM_MPS_CLOSED}), we find the scaling
\begin{align}
\label{EQ_DEFF_RTT}
    \frac{1}{D_\text{eff}^\text{RTT}} &\sim \frac{1}{a^{n}}\left( \frac{b+1}{b + \frac{1}{a b d}} \right)^n \geq  \frac{1}{a^{n}}\left( \frac{b+1}{b + \frac{1}{2 b}} \right)^n \ , 
\end{align}
where we have used an approximation valid for $n \gg b$ in the first step and provided a lower bound by setting $a=2,d=1$ (no bulk legs, minimal physical dimension) 
within the denominator.
We will now consider RTNs on $\{p,q\}$ tilings (with $p>3$) and show that their inverse effective dimension always is upper-bounded by \eqref{EQ_DEFF_RTT}. 
We begin with the simpler $q=3$ case, described by a dual triangular bulk lattice. Via the recursive upper bound \eqref{EQ_Z_BOUND_UPPER}, we can bound the system's partition function by iteratively removing vertices (and its connected edges) from the graph, each removal adding another factor $1+\frac{1}{b^m}$, $m$ being the number of edges attached to the vertex at that step. For a $\{p,3\}$ tiling to be hyperbolic we require $p\geq 7$, which leads the following geometrical insight: The vertices of the outermost layer are still connected to at least $\lfloor \frac{p}{2} \rceil = 3$ other vertices, and as we start to ``strip away'' this layer, each removed vertex still has two remaining connected edges. This logic only breaks down for the last two vertices of the entire graph (where one and zero edges remain, respectively), thus leading to an upper bound 
\begin{equation}
    Z^{\{p,3\}} \leq 2 \left( 1 + \frac{1}{b} \right) \left( 1 + \frac{1}{b^2} \right)^{n_v-2} \ ,
\end{equation}
Inserted into the definition \eqref{EQ_EFF_DIM} of the inverse effective dimension, we thus find the upper bound 
\begin{align}
    \frac{1}{D_\text{eff}^{\{p,3\}}} \lesssim \frac{1}{a^{n}} \left( \frac{1+\frac{1}{b^2}}{1+\frac{1}{b^p d}} \right)^{n_v} \leq \frac{( 1+\frac{1}{b^2} )^{n_v}}{a^{n}} 
\end{align}
on its scaling with $n$,
where in the last step we have taken the limit of infinite bulk logical dimension $d \to\infty$, leading to the strongest possible scaling with $n_v$.
To relate this to the RTT scaling \eqref{EQ_DEFF_RTT}, we need to relate the number of vertices $n_v$ to the number of physical boundary sites $n$. In the language of quantum codes, $k=n_v$ determines the number of logical qubits, with each qubit (or qudit of dimension $d$) hosted by a single vertex. The ratio between $k$ and $n$ then defines the \emph{rate} associated with the bulk-to-boundary code.
For any $\{p,q\}$ tiling, the rate is given by
\begin{equation}
\label{EQ_RATE_HOLO_PQ}
    r \coloneqq \frac{k}{n} = \frac{1}{\sqrt{(p-2)(p-\frac{2q}{q-2})}} \ ,
\end{equation}
in the asymptotic limit of many layers \cite{Jahn:2025manylogical}. If $p>3$, adding a layer of tensors to a smaller patch of a hyperbolic $\{p,q\}$ tiling always adds fewer logical than physical sites, leading to a rate $r<1$, allowing for an \emph{isometric} code with smaller bulk than boundary dimension. 
It then follows that the scaling of the inverse effective dimension of the $\{p,3\}$ RTN with $n$ can be bounded by
\begin{align}
    \frac{1}{D_\text{eff}^\text{\{p,3\}}} \lesssim \frac{( 1+\frac{1}{b^2} )^{n}}{a^{n}} \leq \frac{1}{a^{n}}\left( \frac{b+1}{b + \frac{1}{2 b}} \right)^n \ ,
\end{align}
where the second inequality holds for any $b \geq 2$, the smallest nontrivial bond dimension. It thus follows that 
\begin{equation}
\frac{1}{D_\text{eff}^\text{\{p,3\}}} \lesssim \frac{1}{D_\text{eff}^\text{RTT}}.
\end{equation}
The case of $\{p,q\}$ tilings with $p,q>3$ follows from a similar geometric argument. As we show in the Appendix \ref{APP:PROOF2}, the partition function can be upper-bounded as
\begin{align}
    Z^{\{p,q\}} &\leq  \left( \left( 1+\frac{1}{b} \right)^{\frac{1}{2}} \left( 1+\frac{1}{b^2} \right)^{\frac{\sqrt{3}-1}{2}} \right)^n \ ,
\end{align}
at sufficiently large $n$.
To prove
\begin{align}
    \frac{1}{D_\text{eff}^{\{p,q\}}} \leq \frac{Z^{\{p,q\}}}{a^{n}} \lesssim \frac{1}{a^{n}}\left( \frac{b+1}{b + \frac{1}{2 b}} \right)^n \ ,
\end{align}
where the first step again involves taking the limit $d\to\infty$ to generate an upper bound, we then merely need to confirm the inequality
\begin{align}
    \left( 1+\frac{1}{b} \right)^{\frac{1}{2}} \left( 1+\frac{1}{b^2} \right)^{\frac{\sqrt{3}-1}{2}} \leq \frac{b+1}{b + \frac{1}{2 b}} \ .
\end{align}
Using numerical methods, one finds that it is satisfied for $b \geq 1.968$, with a strict inequality for larger values of $b$. This confirms that ${1}/{D_\text{eff}^{\{p,q\}}} \lesssim {1}/{D_\text{eff}^\text{RTT}}$ for any RTN on a hyperbolic $\{p,q\}$ tiling with $p>3$, assuming that the tiling can be decomposed into layers of vertices. 
\qed

Note that this proof strategy fails for the $p=3$ case as it leads to rates $r>1$, leading to a much faster scaling of the partition function with $n$, potentially producing a larger inverse effective dimension that a corresponding RTT with the same bond dimension. For the same reason, \emph{flat} graphs (e.g., $p=q=4$, as in \eqref{EQ_Z_44_BOUND_LOWER} and \eqref{EQ_Z_44_BOUND_UPPER}) fail to produce a similar bound, as $n_v \sim n^2$.

Having upper-bounded the inverse effective dimension of hyperbolic RTNs by the result for RTTs, we can now also prove that the former equilibrate:
\begin{cor}[Hyperbolic RTN equilibration]
\label{COR_HYP_RTN_EQUI}
    In the scaling limit of infinitely many boundary sites $n$, an RTN state on a regular hyperbolic $\{p,q\}$ tiling with $p>3$ equilibrates perfectly, i.e.,
    \begin{equation}
        \lim_{n \to \infty} \frac{1}{D_\text{eff}^{\{p,q\}}} = 0 \ ,
    \end{equation}
    for any (finite or infinite) choice of bond and logical dimensions $b,d \geq 1$ and a non-trivial physical dimension $a \geq 2$.
\end{cor}

\noindent
\textit{Proof. }
Using Lemma \ref{LEM_EFF_DIM_HYP_VS_RTT}, we can upper-bound the inverse effective dimension of a $\{p,q\}$ RTN with that of an RTT for the number of boundary sites $n$ as well as the same choice of physical, bond, and logical dimensions $a,b,d$, as long as $n$ is sufficiently large. This implies that we can write
\begin{equation}
    \lim_{n \to \infty} \frac{1}{D_\text{eff}^{\{p,q\}}} < \lim_{n \to \infty} \frac{1}{D_\text{eff}^\text{RTT}} = 0 \ .   \qed
\end{equation}
Via Cor.\ \ref{COR_EFF_DIM_ABS_LOWER_BOUND}, the same logic can be immediately extended to the case of a single Haar-random tensor, whose states then also achieve perfect equilibration as $n \to\infty$. One may alternatively apply Cor.\ \ref{COR_RTN_EQUI_LIMIT} to reach this conclusion, as the limit of infinitely large physical dimension and infinitely many sites coincide for the case of a single tensor.

Combining Corollary \ref{COR_EFF_DIM_ABS_LOWER_BOUND} and Lemma \ref{LEM_EFF_DIM_HYP_VS_RTT} now allows us to impose a \emph{hierarchy of equilibration} for RTNs on different geometries: The strongest equilibration is achieved by a single $n$-leg random tensor, or in the language of holographic tensor networks, a ``maximal-size black hole''. 
Furthermore, the equilibration of any RTN state is increased through fusion of adjacent tensors into larger tensors, i.e., ``growing the black hole''.
RTNs on regular hyperbolic geometries, expected to describe typical states in a holographic CFT, lead to weaker equilibration than a single random tensor, but to stronger equilibration than a RTT, i.e., a random tensor network without ``bulk tensors''.

Though we have only proven that this last separation is strict at large $n$, an explicit calculation of $1/D_\text{eff}$ for small tensor networks (for which $Z_{\log\sqrt{b}}$ can still be quickly evaluated in computer-algebra systems) confirms this separation even for small systems: As shown in Fig.\ \ref{fig:ising-tn-geo_eff-dim} for systems of $n=20$ physical sites, the inverse effective dimension $1/D_\text{eff}$ (plotted with an additional factor $a^{n}$ common to all geometries) strictly follows the hierarchy laid out above: The inverse effective dimension for the RTT in subfigure (a) is orders of magnitude above any regular geometry for small bond dimensions $b \leq 10$, while the black hole tensor in subfigure (e) provides a strict lower bound to all other geometries. 
We also observe that there is no strict separation between flat and hyperbolic RTNs (subfigures (b) and (c)), with the latter exhibiting stronger equilibration only at very small bond dimension.

\begin{figure}[t]
    \centering
    \includegraphics[width=1\linewidth]{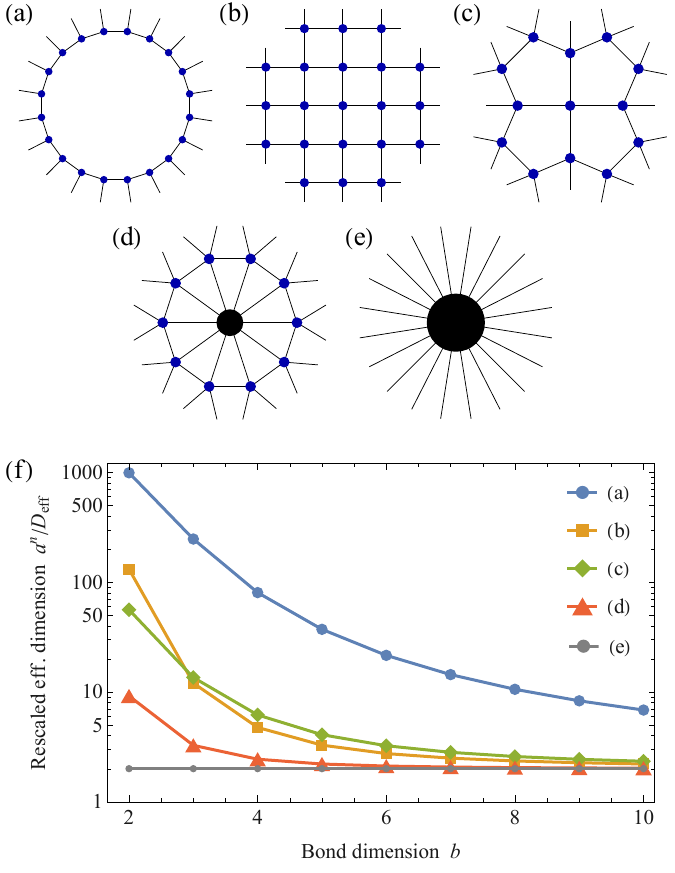}
    \caption{Examples of various random tensor network states with $n=20$ boundary sites and one projected bulk leg per blue tensor. 
    (a) Matrix-product/tensor-train state of a closed chain of tensors.
    (b) Circular cut of a square lattice.
    (c) Hyperbolic lattice (four pentagons of $\{5,4\}$ tiling)
    (d) High-dimensional ``black hole'' tensor at the center of a hyperbolic lattice.
    (e) Single random tensor (``pure black hole'').
    (f) Inverse effective dimension $1/D_\text{eff}$ for the geometries (a)-(e), rescaled by a factor $a^{n}$, where $a$ is the physical dimension per boundary site (set equal to the bond dimension $b$) and $n$ is the number of boundary sites.
    }
    \label{fig:ising-tn-geo_eff-dim}
\end{figure}

\begin{figure}[t]
    \centering
    \includegraphics[width=0.95\linewidth]{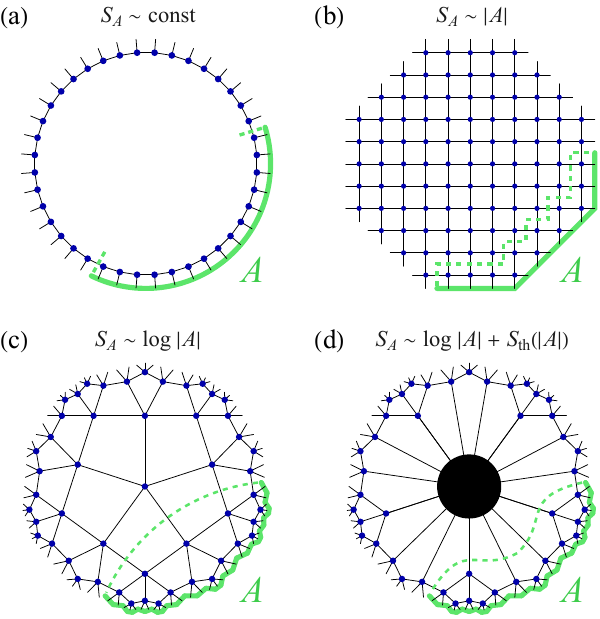}
    \caption{The entanglement entropy $S_A$ of a tensor network's boundary subregion $A$ (green) is upper-bounded by the length of a minimal ``Ryu-Takayanagi'' cut $\gamma_A$ (dashed green), with a scaling $S_A \sim |\gamma_A|$ at large bond dimension \cite{Hayden_2016}.
    (a) For an MPS/TTS geometry, $|\gamma_A|$ does not scale with $|A|$ (1D area law).
    (b) For a flat bulk geometry, $|\gamma_A|$ for regions $A$ with $|A|\ll n$ follows a volume law.
    (c) For a hyperbolic geometry, $|\gamma_A|$ is logarithmically smaller than $A$ (1D area law with logarithmic corrections).
    (d) Fusing central tensors of a hyperbolic geometry into a single random (black hole) tensor creates an additonal ``thermal'' (or ``horizon'') contribution $S_\text{th}(|A|)$ for sufficiently large $A$. 
    }
    \label{fig:rt-cuts}
\end{figure}

The implications of choosing different RTN geometries on the physical states can be well expressed in terms of the scaling of entanglement entropies $S_A$ for subregions $A$ of the boundary sites \cite{arealaw,Vasseur_2019}.
For any tensor network of constant bond dimension $b$, it is well-known to be upper-bounded as 
\begin{equation}
\label{EQ_RT_TN}
    S_A \leq |\gamma_A| \log b \ ,
\end{equation}
where $|\gamma_A|$ is length of a minimal cut $\gamma_A$ through the tensor network (the number of edges through which $\gamma_A$ cuts) with matching endpoints $\partial A = \partial \gamma_A$. This notation is reminiscent of the Ryu-Takayanagi formula \cite{Ryu_2006,Ryu_2006long} with a replacement $\log b \to 1/4 G$, $G$ being the gravitational constant of the holographically dual bulk spacetime, and the inequality replaced by an equality up to $O(G^0)$ corrections. Such a similarity was first noted in Ref.\ \cite{Swingle:2009bg} for the MERA tensor network and sharpened for RTNs in Ref.\ \cite{Hayden_2016}: At large $b$, the inequality \eqref{EQ_RT_TN} then saturates, with $\gamma_A$ taking the form of a domain wall in an equivalent Ising spin description of the bulk sites.
The form of $\gamma_A$ for different TN geometries is shown in Fig.\ \ref{fig:rt-cuts}. We can see that RTT/MPS geometries (subfigure (a)) are only suitable to describe 1D quantum states up to a \emph{strict area law}, meaning that a 1D region $A$ with a 0D boundary $\partial A$ (the endpoints) can only have entanglement $S_A \sim \text{const}$ that does not scale with $A$ \cite{arealaw}. In order to approximate quantum states with large entanglement for large regions well, such a TN would therefore be an inefficient ansatz requiring a bond dimension scaling with the total system size.
Fig.\ \ref{fig:rt-cuts}(b) shows a TN geometry with a flat 2D bulk, allowing up to \emph{volume-law} entanglement $S_A \sim |A|$. Note that not all tensor networks on such flat geometries describe volume-law entangled states: For example, using free-fermion \emph{matchgate} tensors typically leads to gapped phases with $S_A \sim \text{const}$ and two-point correlations decaying exponentially with distance \cite{Jahn:2017tls}.
A third choice of a TN bulk geometry is given by hyperbolic regular tilings (subfigure (c)), which allow for an area law with logarithmic corrections $S_A \sim \log |A|$. Such hyperbolic bulk geometries are suitable to describe ground states of critical models, both because of the correct entanglement scaling \cite{Calabrese:2004eu} and because the graph distance $\delta$ between two points with boundary distance $d$ will scale as $\delta \sim \log d$, producing a decay of boundary two-point correlations that is polynomial in $d$ if bulk correlations decay exponentially in $\delta$ \cite{Hayden_2016,Jahn:2017tls}.
Finally, adding a ``black hole'' to a hyperbolic TN geometry, i.e., fusing together some of its tensors into one larger tensor, adds additional contributions $S_\text{th}(|A|)$ to the entanglement entropy. These ``thermal'' contributions become significant when $|A|$ becomes comparable to the system size and $\gamma_A$ is deformed by the black hole, partially wrapping around its horizon (which is considered a thermal object in gravity and holography), analogous to actual black hole insertions in AdS/CFT \cite{Ryu_2006}.

\subsection{Circular ensemble dependence of equilibration and delocalization}
In general, one would like to derive statistical mechanics for states which preserve some symmetry; in the case of RTN ensembles, such symmetries may be expressed in terms of the ensemble from which each local tensor is chosen. 
One way of doing this is by constructing random tensor networks not just from the \emph{circular unitary ensemble} (CUE), which is often simply referred to as ``Haar-random'', but the other circular ensembles as well. The \emph{circular orthogonal ensemble} (COE) in particular has gained recent interest in quantum gravity \cite{harlow2025quantummechanicsobserversgravity, Akers:2025ahe}. A detailed exposition on Dyson's circular ensembles can be found in Ref.\ \cite{forrester2010loggases}. For the COE we have 
\begin{align}
\label{weingartenorth2n4}
    & \int dO O_{i_1,j_1}O_{i_2,j_2} = \frac{1}{d}\delta_{i_1,i_2}\delta_{j_1,j_2}, \\
    & \int dO O_{i_1,j_1}O_{i_2,j_2}O_{i_3,j_3}O_{i_4,j_4} =  \nonumber \\
    & A\delta_{i_1,i_2}\delta_{i_3,i_4}\delta_{j_1,j_2}\delta_{j_3,i_4} + A\delta_{i_1,i_3}\delta_{i_2,i_4}\delta_{j_1,j_3}\delta_{j_2,j_4} \nonumber \\
    & + A \delta_{i_1,i_4}\delta_{i_2,i_3}\delta_{j_1,j_4}\delta_{j_2,j_3} + B \delta_{i_1,i_3}\delta_{i_2,i_4}\delta_{j_1,j_2}\delta_{j_3,j_4} \nonumber \\
    & + B \delta_{i_1,i_4} \delta_{i_2,i_3}\delta_{j_1,j_3}\delta_{j_2,j_4} + B \delta_{i_1,i_2} \delta_{i_3,i_4}\delta_{j_1,j_4}\delta_{j_2,j_3} \nonumber \\
    & + B \delta_{i_1,i_4} \delta_{i_2,i_3}\delta_{j_1,j_2}\delta_{j_3,j_4} + B \delta_{i_1,i_2} \delta_{i_3,i_4}\delta_{j_1,j_3}\delta_{j_2,j_4} \nonumber \\
    & + B \delta_{i_1,i_3} \delta_{i_2,i_4}\delta_{j_1,j_4}\delta_{j_2,j_3}, \nonumber
\end{align}
where $A \coloneqq \frac{d+1}{d(d+2)(d-1)}$ and $B\coloneqq -\frac{1}{d(d+2)(d-1)}$. With post-selection, we have 
\begin{align}
    &\int dO O_{i_1,j_1}O_{i_2,j_2}O_{i_3,j_3}O_{i_4,j_4} \ket{0}_{j_1 j_2 j_3 j_4} = \\
    & = (A + 2B) (\delta_{i_1 ,i_2}\delta_{i_3, i_4} + \delta_{i_1, i_3}\delta_{i_2, i_4} + \delta_{i_1, i_4}\delta_{i_2, i_3} ) \nonumber \\ 
    & = \frac{1}{d(d+2)} (\delta_{i_1 ,i_2}\delta_{i_3, i_4} + \delta_{i_1, i_3}\delta_{i_2, i_4} + \delta_{i_1, i_4}\delta_{i_2, i_3} ) .\nonumber
\end{align}
Contrasting this with the CUE Haar-random unitary case, for which we simplify (\ref{fullcalc}), to get
\begin{align}
    & \int dU U_{i_1,j_1}U_{i_2,j_2}U_{i_3,j_3}U_{i_4,j_4} \ket{0}_{j_1 j_2 j_3 j_4} \\
    \nonumber
    &= \left(1- \frac{1}{d} \right) \left(\frac{1}{d^2-1} \right)(\delta_{i_1, i_3}\delta_{i_2, i_4} + \delta_{i_1, i_4}\delta_{i_2, i_3}) \\
    \nonumber
    &=  \frac{1}{d(d+1)}(\delta_{i_1, i_3}\delta_{i_2, i_4} + 
    \delta_{i_1, i_4}\delta_{i_2, i_3})
    .
    \nonumber
\end{align}
By applying this to a simple tensor network of the type considered in figure \ref{fig:equilibration-intro} and evaluated for the CUE in (\ref{fullcalc}), we find $\frac{2b^2+2b}{q^2(q+1)^2}$ and $\frac{3b^2+6b}{q^2(q+2)^2}$ for the COE (note again that these are upper bounds assuming product bulk and boundary states). The values are comparable and we expect this to generalize (though we do not prove it) for the other CUE based tensor network geometries we considered in this work. 
For the COE case, one would then need the partition function onto a 3-spin classical model. 

How will equilibration and localization change in the case of the symplectic ensemble? We can make use of the following theorem: 
\begin{thm}[\cite{west2024randomensemblessymplecticunitary}] The circular symplectic ensemble and the circular unitary ensemble on states are indistinguishable i.e. 
    \begin{equation}
        \mathbb{E}_{Sp}[U^{\otimes t}\ket{\psi_0}\bra{\psi_0}^{\otimes t}(U^\dagger)^{\otimes t}] = \mathbb{E}_{U}[U^{\otimes t}\ket{\psi_0}\bra{\psi_0}^{\otimes t}(U^\dagger)^{\otimes t}]
    \end{equation}
    meaning that symplectic random states form unitary state t-designs for all t.
\end{thm}
A corollary of this theorem is the following.

\begin{cor}[Symplectic unitaries are generic for random tensor networks]
    Random tensor networks generated by using symplectic unitaries will have an equilibration value equal to random tensor networks generated by random unitaries. 
\end{cor}

\subsection{Equilibration of n-point functions}
We can extend the relationship between late-time observable fluctuations and $D_\text{eff}$ from Lemma \ref{fluclem} to also account for multi-point observables. 
We consider observables $B=A_1 A_2 \dots A_n$ where each individual observable is a simple operator in norm, i.e., $\norm{A_i}=O(1)$. Then $\norm{B} \leq \prod_i \norm{A_i}$ at all times, as Heisenberg evolution is norm-preserving. The temporal fluctuations are then given by 
\begin{equation}
    (\Delta B_\psi^\infty)^2 = \lim_{\tau \xrightarrow{} \infty} \frac{1}{\tau} \int_0^\tau (\bra{\psi}B(t)\ket{\psi} - B_\psi^\infty)^2
\end{equation}
then similar to Lemma \ref{fluclem} we take a state vector $\ket{\psi}=c_j \ket{j}$ and 
\begin{align}
\label{EQ_EFF_DIM_N_POINT}
    (\Delta B_\psi^\infty)^2 &\leq \norm{B}^2 \sum_j |c_j|^4 \\
    \nonumber
    &= O \left(\frac{\prod_l \norm{A_l}^2}{D_{\text{eff}}} \right) \\
     \nonumber
    &= O \left(\frac{1}{D_{\text{eff}}} \right)
     \nonumber.
\end{align}
This implies that $1/D_\text{eff}$ is still a good measure for the fluctuations of $n$-point observables as long as these are of sufficiently low ``complexity'', i.e., much more localized than the state vector. For more complex observables, the bound \eqref{EQ_EFF_DIM_N_POINT} will become increasingly loose, i.e., fluctuations of observables can be much smaller than $1/D_\text{eff}$. Another way of expressing the same ides is that more complex operators can increase the dimensionality of the state space in which asymptotic time evolution occurs.

\section{Discussion}
\label{SEC_DISCUSSION}

\subsubsection*{Many-body physics}
\label{SUBSEC_DISC_MANY_BODY}

Considering our results for the dependence of equilibration in RTN states on different geometries, we observe a striking match between the hierarchy of equilibration for RTNs and the hierarchy of entanglement scaling for subregions: Though the inverse effective dimension $1/D_\text{eff}$ is a global quantity of RTNs, it appears to scale with the typical entanglement of subregions:

\begin{con}[Growing equilibration with subregion entanglement entropy]
\label{conject1}
    The strength of equilibration of simple operators applied to a typical state in a quantum many-body phase increases with the scaling of subregion entanglement entropy in that phase.
\end{con}
This conjecture arises from our observation that an RTN geometry that leads to less subregion entanglement appears to decrease the effective dimension and hence increase late-time fluctuations of observables.
Intriguingly, the effective dimension explored by a time-evolved local operator on a single \emph{sample} of an RTN thus appears to scale with the dimension of the RTN \emph{ensemble} as a whole: For example, an RTT ensemble is restricted to samples from an area-law-entangled subspace of physical states, whereas a single random tensor produces an ensemble in the much larger space of volume-law states (essentially the full Hilbert space). And our calculations show that the effective dimension of states from the latter ensemble is provably larger than the former.
This apparent relationship between the effective dimension and an ``entropic'' counting of degrees of freedoms for subregions (i.e., $e^{S_A}$) would also explain the switch in hierarchy between flat and hyperbolic RTNs observed for small systems in Fig.\ \ref{fig:ising-tn-geo_eff-dim}(b-c): While at large bond dimension the geometric argument above implies a larger (volume-law) entanglement entropy for the flat RTNs, at small bond dimension these RTNs may be closer to describing a gapped phase with a strict area law, i.e., with smaller entanglement than the hyperbolic RTNs, as observed in generic (non-random) matchgate tensor networks \cite{Jahn:2017tls}. A link between our equilibration results and different scaling laws for entanglement entropies may arise from further studies of wavefunction delocalization in such phases: For example, volume-law-entangled phases generally exhibit stronger delocalization than area-law ones.

Our results can also be considered within the context of \emph{eigenstate thermalization hypothesis} (ETH) \cite{Srednicki_1994}, which is expected to describe certain operators in a holographic CFT \cite{lashkari2016eigenstatethermalizationhypothesisconformal, saad2019latetimecorrelationfunctions, Pollack_2020, Bao_2019, jafferis2023jt, jafferis2023matrix, Sonner_2017, Nayak_2019, Dymarsky_2018, Dymarsky:2016ntg, Basu_2017, Lashkari:2017hwq, Faulkner:2017hll, Brehm:2018ipf, Romero-Bermudez:2018dim}. ETH is a mathematical ansatz for the matrix elements of simple operators $\mathcal{O}$ in the eigenstate basis of a chaotic Hamiltonian and is given by 
\begin{equation}
    \bra{E_m}\mathcal{O}\ket{E_n} = \overline{O}(E)\delta_{m,n} + e^{-S(\overline{E})/2}f_O(E,\omega)R_{m,n} \ ,
\end{equation}
where $\overline{O}(E)$ the microcanical expectation value $f_O(E,\omega)$ a smooth spectral function of the averaged energy $E = (E_m + E_n)/2$ and the energy difference $\omega = E_m -E_n$ satisfying $f_O(E,-\omega)=f_{O}(E,\omega)$ for real Hamiltonians $R_{m,n}$ a random variable with zero mean, unit variance satisfying $R_{m,n}=R_{n,m}$. The thermodynamic entropy (also called the microcanonical entropy) $S(E)$ is defined as the logarithm of the coarse-grained density of states, i.e., 
the number of eigenstates of energy $E$ is given by $e^{S(E)}$. For a system that satisfies ETH, the fluctuations around the infinite-time average are suppressed as $e^{-S(E)}$. This can be seen from part of Lemma \ref{fluclem}, where the fluctuations $\Delta A_{\psi}^{\infty}$ are given by
\begin{align}
    \sum_{j \neq k} |c_j|^2 |c_k|^2 |A_{j,k}|^2 \leq \max |A_{j,k}|^2 \propto e^{-S(E)}.
\end{align}
Our contribution has been to prove equilibration and bound the fluctuations precisely in a model without assuming the ETH. If the random tensor network states are to be interpreted in an ETH setting, then since $S(E) \propto \ln(\rho_D)$ where $\rho_D$ is density of states, then it follows that ${\mathrm{ln}}(\Delta A_\psi^\infty) \propto -\ln(\rho_D)$. We can thus interpret our computations not only as fluctuations around a thermal value but also as computing the density of states for a system satisfying ETH. This essentially implies that systems with smaller fluctuations around the thermal equilibrium value will have a higher density of states. We are for this reason led to make Conjecture \ref{conject2}.
\begin{con}[Role of chaotic systems with higher microcanonical entropy]
\label{conject2}
    Chaotic systems with higher microcanonical entropy will have smaller fluctuations around the thermal equilibrium value.
\end{con}
The conjecture would be true for systems which satisfy ETH by the following argument: systems with higher density of states will have a higher thermodynamic entropy $S(E)$, since the fluctuations for systems satisfying ETH are proportional to $e^{-S(E)}$, it immediately follows that the fluctuations will be smaller for systems with higher microcanonical entropy.

One can contrast our results for random tensor networks to stabilizer states, in particular one may wish to contrast the value of the quantity $|c_i|^4$, with the HaPPY code \cite{Pastawski:2015qua}. Of course, typicality arguments cannot be made for stabilizer states like they can for random states, the quantity, however, will correspond to some measure of delocalization. For qubits, the stabilizer state vector representation \cite{Dehaene_2003, gross2007lulcconjecturediagonallocal} will be given by 
\begin{equation}
    \ket{K,q,b} := \frac{1}{|K|}\sum_{x \in K}i^{b x}\omega^{q(x)}\ket{x} 
\end{equation}
where $K\subset \mathbb{F}_2^n$, $q$ is a quadratic form on $K$ and $b \in \mathbb{F}_2^n$ and $n$ is the number of qubits. Stabilizer states will have wavefunctions that will be more localized than random states. 

We have shown how the most elementary symmetries associated to the different circular ensembles can be incorporated in our framework. However, one may press on and ask instead about the equilibration and localization properties of topological phases of matter. The action of symmetries on a quantum-mechanical system leads to a decomposition of the Hilbert space in terms of superselection sectors. In this sense, one can think of common representatives of phases of matter which respect a certain symmetry whereby one considers states where the randomization is taken within a superselection sector \cite{Piroli_2020}. Such considerations will also be important for defining typical instances of topological phases of matter. Such systems would be interesting in their own right as instances of a correspondence between boundary global symmetry and bulk gauge symmetry in line with expectations from the holographic duality. Such models may shed light on confinement/deconfinement transitions in gauge theories and allow for a testable model to relate this to the physics of localization-delocalization transitions. Some models of this type are described in \cite{qi2022emergentbulkgaugefield, dong2023holographictensornetworksbulk, Akers_2024a, Akers_2024b}

Along similar grounds, it makes sense to apply our setup to dual-unitary circuits \cite{bertini2025exactlysolvablemanybodydynamics}. The equilibration quantity is also the inverse participation ratio. Localization-delocalization transitions can be probed by finding a suitable trajectory in the space of dual-unitaries. We plan to pursue this in future work \cite{gemma}.  
\subsubsection*{Holography and conformal field theory}
\label{SUBSEC_DISC_HOLO}

The RTNs generally considered as models of AdS/CFT are those with hyperbolic bulk geometry (see, e.g., Fig.~\ref{fig:ising-tn-geo_eff-dim}(c) and Fig.\ \ref{fig:rt-cuts}(c)). AdS/CFT exhibits a semi-classical limit at large $N$, where $N$ counts local boundary degrees of freedom (the rank of the gauge group) while also determining the bulk's gravitational constant $G \sim {1}/{N^2}$ (using $\hbar=c=1$). In RTN models of holography, $N$ is typically identified with the boundary physical dimension $a$, which is assumed to be equal to the bond dimension $b$ \cite{Hayden_2016}. Such an identification is motivated by the saturation of the tensor-network Ryu-Takayanagi formula \eqref{EQ_RT_TN} in this limit, matching what occurs in the large $N$ limit of holography: Quantum fluctuations in the bulk and hence back-reaction on the geometry are strongly suppressed, leading to a geometry in which the area of the RT surface $\gamma_A$ becomes approximately classical or ``scalar'', independent of the bulk state \cite{Faulkner:2013ana,Engelhardt:2014gca}.
While it was already known that fluctuations around operator expectation values of the RTN ensemble disappear in the limit of large bond dimension \cite{Hayden_2016}, our Cor.\ \ref{COR_RTN_EQUI_LIMIT} now implies that asymptotic fluctuations of \emph{time-dependent} observables disappear as well. This should not be surprising, as in this limit bulk gravitational fluctuations also disappear while the notion of a boundary Hamiltonian becomes ill-defined.
At finite but large $N$, however, we expect the time evolution of local boundary operators to correspond to bulk dynamics with small (perturbative) quantum-gravitational effects.
If RTNs are a good model of holography in this regime, we would expect $D_\text{eff}$ to provide a measure of the size of the physical Hilbert space in which such dynamics can occur. 
For saturating the upper bound \eqref{EQ_EFF_DIM} for $1/D_\text{eff}$, we require that the bulk state vector $\ket{\Phi}$ and the boundary reference state vector $\ket{\phi}$ are close to product states, i.e., have small entanglement. These are reasonable assumptions in the holographic context: A highly entangled bulk state would corresponds to a strong deformation on the geometry (and possibly even topology) far from the (finite-$N$) vacuum, while a boundary Hamiltonian for a $1+1$-dimensional holographic CFT should be spatially local and hence have eigenstates with no more that area-law entanglement.
The inverse of the bound \eqref{EQ_EFF_DIM} thus determines the minimal $D_\text{eff}$ associated with an RTN state.
Assuming then that this quantity counts degrees of freedom of holographic bulk dynamics, we find this counting to be consistent with gravitational effects: Replacing a bulk region by a black hole, i.e., fusing a set of tensors with some nontrivial geometry into a single random tensor, increases $D_\text{eff}$. This is consistent with a degree-of-freedom counting in the presence of gravity, where the maximum entropy of a spatial region with boundary area $A$ is that of a black hole with $S_\text{BH}={A}/{(4G)}$ \cite{Bekenstein:1973ur,Hawking:1974sw} and horizon area $A$. Concretely, the effective dimension of such a black hole tensor takes the form \eqref{EQ_EFF_DIM_MIN}, which for a large ``horizon'' dimension $\dim H_\text{hor} = a^n$ takes the simple form $D_\text{eff} \approx \frac{1}{2} \dim H_\text{hor}$.
In this picture, the lowest-energy ensemble we can prepare with an RTN preserves the maximum number of discrete hyperbolic symmetries, i.e., has the geometry of a uniform, regular hyperbolic tiling. Such an ensemble should therefore describe finite-$N$ fluctuations with sufficiently low energy (i.e., small back-reaction) to approximately preserve the RTN geometry. In semiclassical AdS spacetime, field excitations generally lead to the formation of black holes \cite{Bizon:2011gg}, so $D_\text{eff}$ could be interpreted as measuring the horizon area of an initial configuration time-evolved to asymptotically late times. RTNs whose bulk geometry is deformed by a fusion of tensors would describe higher-energy initial ensembles that already contain black holes, evolving towards asymptotic configurations with larger horizons and higher $D_\text{eff}$. An interesting future direction would be to explore how positive- and negative-curvature RTN deformations, as well as contributions from highly entangled bulk states, affect $D_\text{eff}$ and whether these results are consistent with a more careful degree-of-freedom counting in AdS/CFT.
An interesting aspect of our small-system calculations in Fig.~\ref{fig:ising-tn-geo_eff-dim}(b-c) is that there exists a bond dimension value for which the value of $D_\text{eff}$ coincides for flat and hyperbolic geometries. This may be a signature of an emergent Weyl invariance when viewing the RTNs as a path integral preparing an ensemble of boundary states, as previously observed in the context of holographic path-integral complexity \cite{Caputa:2017urj,Caputa:2017yrh}.

Our results display the versatility of tensor network models in describing aspects of holographic CFTs. In particular, these CFTs are expected to satisfy a version of the ETH \cite{chen2024holographicrenyientropy2d}. The notion of ETH is associated with the thermodynamic limit. The standard thermodynamic limit is one reached by taking the limit of infinite system size $L \xrightarrow{} \infty$, however, there is an ``internal" thermodynamic limit for CFTs which is obtained in the large central charge limit $c\to\infty$. This is a necessary condition for the theory to have a weakly coupled gravity dual through AdS/CFT in which the phenomenon of thermalization is related to black hole formation and evaporation. The two thermodynamic limits can be taken simultaneously for $L \gg \beta$ which is dual to high-temperature black holes. 
RTNs similarly allow for two types of limits: One in which the total size of the TN geometry diverges, and another in which the physical and bond dimensions $a$ and $b$ diverge, with the latter describing an infinite number of local degrees of freedom similar to the $c\to\infty$ limit. As we have shown, RTNs with a hyperbolic geometry generally lead to a divergent $D_\text{eff}$ and hence perfect equilibration in both limits. 

Tensor networks have been shown to display a flat entanglement spectrum which corresponds to fixed area states in gravity. The situation can be remedied by introducing link states. The extension of our results to random tensor networks with link states will be an important step towards displaying equilibration for a larger class of states which behave like CFT states. The non-EPR contractions in this case will be upper bounded by the EPR contractions due to the Schmidt decomposition of non-maximally entangled states. 

The large $N$ limit of RTNs and equilibration therein is another subtle issue. Local operator algebras of boundary of infinite tensor networks have been studied using inductive limit algebras \cite{chemissany2025infinitetensornetworkscomplementary}. The work has been restricted to perfect tensors, however. RTNs in the infinite bond dimension limit have entanglement properties very similar to perfect tensors. Extensions of the inductive limits to RTNs will prove useful for characterizing the operator algebras of phases of matter and holographic systems in probabilistic/ensemble settings. 

It is worth mentioning that gravity behaves rather similarly to hydrodynamics \cite{banks2025hydrodynamicapproachquantumgravity,hubeny2011fluidgravitycorrespondence}. In this work, we derived statistical mechanics in a sense. Deriving hydrodynamics \cite{Banks_2019, wang2025eigenstatethermalizationhypothesiscorrelations, Doyon_2022} of a specific type may very well be crucial for generating an effective gravitational description of holographic tensor networks. Some steps have been made in the direction of describing operator spreading of random circuits in terms of hydrodynamics \cite{Nahum_2017,Nahum_2018, Rakovszky_2018}, however, a full picture is still lacking. The hydrodynamic description of many-body systems is one where a systems evolution is described in terms of emergent macroscopic degrees of freedom. The validity of such an approximation is contingent on properties of macroscopic observables like density and velocity. Holography can then be understood as a hydrodynamic description of boundary entanglement where the entropies are its macroscopic phase space \cite{Bao_2015}.

There are other intriguing directions into which our work could be extended. Here we only considered RTNs with boundary theories in one spatial dimension, but in principle one can consider boundary and bulk systems in any dimension. Some of our results, in particular Lemma \ref{LEM_RTN_EFF_DIM_GROWTH} for vertex fusion, already apply to RTNs on arbitrary graphs. It seems plausible that such higher-dimensional RTNs can shed light on the equilibration behavior of critical holographic theories in higher dimensions.

Another possible direction concerns the connection between RTN ensembles and typicality in gravity.
We saw that at large bond dimension, many classes of RTNs show asymptotic equilibration, meaning that their boundary theories can be described by simple quantum mechanics for typical states, consistent with standard holography. However, bulk gravity can in some settings (such as $1+1$D Jackiw-Teitelboim (JT) gravity) be dual to a boundary \emph{ensemble theory} \cite{saad2019semiclassicalrampsykgravity}. It would be interesting to explore settings of non-equilibrating RTNs (such as those with flat-space bulk geometry) whose boundary ensemble cannot be described by a typical state. These might describe settings of holography where the boundary theory has a fundamental ensemble nature.

In our study of equilibration we also restricted ourselves to observables that are of low complexity, with quantities such as the effective dimension $D_\text{eff}$ being in principle accessible to a low-complexity observer. This could have a bearing on recent studies of complexity constraints on observers in holographic bulk spacetimes, in particular those whose complexity is constrained by a \emph{non-isometric} encoding \cite{akers2022blackholeinteriornonisometric,harlow2025quantummechanicsobserversgravity, Akers:2025ahe}. In such settings, $D_\text{eff}$ may be related to the dimension of the fundamental Hilbert space onto which the non-isometric code projects.

Finally, we would like to point out an interesting relation that is known in random matrix theory. The eigenvalue distribution of CUE ensembles is the same as the partition function obtained from Chern-Simons theories as a result of the Coulomb-gas formalism \cite{forrester2010loggases}. The Coulomb-gas formalism has previously been exploited in conformal minimal models and studies of gravity \cite{Pi_tek_2022}. We discuss part of the connection in \ref{APP:coulombgas}. It would be interesting to investigate whether these results can have potential bearings on random tensor networks and random circuit models of black holes \cite{akers2022blackholeinteriornonisometric,harlow2025quantummechanicsobserversgravity,Akers:2025ahe}. \\

\section{Acknowledgments}
We would like to warmly thank Chris Akers, Alexander Altland, Dmitry Bagrets, Christian Bertoni, Lennart Bittel, Nele Callebaut, ChunJun Cao, Sebastian Diehl, Antonio Anna Mele, Silvia Pappalardi, Jason Pollack, Leo Shaposhnik, Spyros Sotiriadis, Tadashi Takayanagi, 
Maksimiliam Usoltcev, Konstantin Weisenberger, 
Zhuo-Yu Xian for very inspiring discussions. This work has been supported by the Einstein Research Unit on Quantum Devices, Berlin Quantum,
the DFG (CRC 183 and FOR 2724), the Clusters of Excellence MATH+ and ML4Q, the BMFTR (MuniQC-Atoms), and the 
European Research Council (DebuQC).

\appendix

\section{Ensembles and typicality in statistical mechanics}
\label{APP:DETEQUIB}

Here we provide a deeper background into ensembles and typicality, largely following the definitions and notations of Ref.~\cite{Gogolin_2016}.
In quantum statistical mechanics, one defines the \emph{micro-canonical ensemble} and \emph{state} with respect to an energy interval $[E, E+ \Delta]$. The micro-canonical state to any subset $R \subseteq \mathbb{R}$ of the real numbers of a system with Hilbert space $\mathcal{H}$ and Hamiltonian $H \in \mathcal{O}(\mathcal{H})$ with spectral decomposition $H= \sum_{k=1}^{d'}E_k \Pi_k$ is defined as 
\begin{equation}
\label{microcanonical}
    \sqcap[H](R) := \frac{\sum_{k:E_k \in R} \Pi_k}{Z_{\text{mc}}[H](R)} \in \mathcal{S}(\mathcal{H}),
\end{equation}
where $Z_{\text{mc}}[H]$ is the micro-canonical partition function defined as 
\begin{equation}
    Z_{\text{mc}}[H](R) := \Tr(\sum_{k:E_k \in R}\Pi_k).
\end{equation}
Having defined the micro-canonical ensemble, we move onto typicality which replaces the postulate that systems are described by an ensemble with the theorem that for almost every 
pure state in a large Hilbert space region, 
local and coarse observables already look as if the state had been drawn from that ensemble. The subject of typicality has a long history in statistical mechanics and 
quantum thermodynamics \cite{schrodinger1927energieaustausch, neumann2010proof,lloyd1988phd,lloyd2013purestatequantumstatistical,gemmer2009quantum,goldstein2006canonical,Gogolin_2016,BertoniThermal, Eisert_2015,1112.5295}.

These arguments can be formulated in the context of drawing state vectors uniformly from a high dimensional subspace and having reduced states on small subsystems appear similar to the reduced state of the micro-canonical state corresponding to that subspace. Drawing a state vector uniformly at random from a subspace intuitively means that any state from the subspace is as probable as any other. This can be made mathematically precise by the notion of left/right invariant measures. Haar's theorem implies that for any finite $d$ there is a unique left and right invariance, countably additive, normalized measure on the unitary group $U(d)$ called the Haar measure on $U(d)$ which we denote here by $\mu_{\text{Haar}}[U(d)]$. Left and right invariance simply means that for any unitary $U \in U(d)$ and any Borel set $\mathcal{B} \subseteq U(d)$
\begin{equation}
    \mu_{\text{Haar}}[U(d)](\mathcal{B}) = \mu_{\text{Haar}}[U(d)](U \mathcal{B}) = \mu_{\text{Haar}}[U(d)](\mathcal{B}U),
\end{equation}
where $U \mathcal{B}$ and $\mathcal{B}U$ are the left and right translations of $\mathcal{B}$. In this sense, the Haar measure $\mu_{\text{Haar}}[U(d)]$ is the uniform measure on $U(d)$. 
The Haar measure on the group of unitaries that map a restricted subspace $\mathcal{H}_R \subseteq \mathcal{H}$ of dimension $d_R$ into itself induces in a natural way, a uniform measure $\mu_{\text{Haar}}[\mathcal{H}_R]$ on state vectors $\ket{\psi} \in \mathcal{H}_R$. We call state vectors and pure quantum state vectors $\ket{\psi}\bra{\psi}$ drawn according to this measure, Haar random and write $\ket{\psi} \sim \mu_{\text{Haar}}[\mathcal{H}_R]$. 
A practical way to obtain state vectors distributed according to this measure is to fix a basis $\{\ket{j}\}_{j=1}^{d_R}$ for the subspace $\mathcal{H}_R$ and then draw the real and imaginary part of $d_R$ complex numbers $\{c_j\}_{j=1}^{d_R}$ from normal distributions of mean zero and variance one. The state vector 
\begin{equation}
    \ket{\psi} = \frac{\sum_{j=1}^{d_R}c_j \ket{j}}{(\sum_{j=1}^{d_R} |c_j|^2)^{1/2}}
\end{equation}
is then distributed according to $\mu_{\text{Haar}}[\mathcal{H}_R]$, 
i.e., $\ket{\psi} \sim \mu_{\text{Haar}}[\mathcal{H}_R]$ \cite{Zyczkowski_2001}. We denote the probability that an assertion $\mathbb{A}(\ket{\psi})$ about a state vector as being true if $\ket{\psi} \sim \mu_{\text{Haar}}[\mathcal{H}_R]$ by $\mathbb{P}_{\ket{\psi} \sim \mu_{\text{Haar}}[\mathcal{H}_R]}(\mathbb{A}(\ket{\psi}))$. Typicality is a consequence of measure concentration and can be made more quantitative by deviation bounds, in particular Levy's lemma as:

\begin{thm}
\label{measureconcentration}
    [Measure concentration for quantum state vectors \cite{Gogolin_2016}] Let $R \subset \mathbb{R}$ and let $\mathcal{H}_R \subseteq \mathcal{H}$ be a subspace of the Hilbert space $\mathcal{H}$ of a system with a Hamiltonian $H \in \mathcal{O}(\mathcal{H})$ which is spanned by the eigenstates of $H$ to energies in $R$ and let $d_R := \dim(\mathcal{H}_R)$. Then for every $\epsilon>0$ it holds that: \\
    1. For any operator $A \in \mathcal{B}(\mathcal{H})$
    \begin{align}
        \mathbb{P}_{\ket{\psi} \sim \mu_{\text{Haar}}[\mathcal{H}_R]}(|\langle A \rangle_{\bra{\psi}\ket{\psi}} - \langle A \rangle_{\sqcap[H](R)}| \geq \epsilon) \quad\nonumber\\
        \leq 2e^{\frac{-C d_R \epsilon^2}{\norm{A}_\infty^2}} \ ,
    \end{align}
    2. For any set $\mathcal{M}$ of positive operator-valued measurements (POVMs)
    \begin{align}
        \mathbb{P}_{\ket{\psi} \sim \mu_{\text{Haar}}[\mathcal{H}_R]}(\mathcal{D}_{\mathcal{M}}(\ket{\psi}\bra{\psi}, \sqcap[H](R)) \geq \epsilon) \quad\nonumber\\
        \leq 2h(\mathcal{M})^2 e^{\frac{-C d_R \epsilon^2}{h(\mathcal{M})^2}} \ ,
    \end{align}
    where $C=\frac{1}{36\pi^3}$ and 
    \begin{equation}
        h(\mathcal{M}) := \min(|\cup \mathcal{M}|,\dim(\mathcal{H}_{\text{supp}(\mathcal{M})})) \ ,
    \end{equation}
    where $\mathcal{D}(\rho,\sigma)$ is the trace distance.
\end{thm}
Physically, 
the theorem says that for a quantum many-body system, a generic pure 
state restricted to a narrow 
energy shell reproduces, with overwhelming probability, 
the micro-canonical expectation values of all observables.

A physically relevant case is when the reduced state of a random state on some small subsystem from the subspace corresponding to some energy interval is indistinguishable (with high probability) from 
the reduction of the corresponding micro-canonical state. 

To be precise, we consider the case when $\text{supp}(\mathcal{M})$ is contained in some small subsystem $S \supseteq \text{supp}(\mathcal{M})$ and $R=[E, E+ \Delta]$ is some energy interval. Then the theorem yields a probabilistic bound on the distance $\mathcal{D}(\ket{\psi}\bra{\psi}^S, \sqcap^S[H](E, E+ \Delta))$. If $\ket{\psi} \sim \mu_{\text{Haar}}[\mathcal{H}_R]$ and the dimension $d_R$ of the micro-canonical subspace $\mathcal{H}_R$ to the energies in the interval $[E, E+ \Delta]$ fulfills $d_R \gg  d_S$, then $\mathcal{D}(\ket{\psi}\bra{\psi}^S, \sqcap^S[H](E, E+ \Delta))$ is small with very high probability. 

The same situation will hold in the more general setting that one only has access only to a sufficiently small number of measurements, which in total have a sufficiently small number of different outcomes. If the total number of different outcomes $|\cup \mathcal{M}|$ is much smaller than the dimension of the subspace corresponding to the energy interval 
$[E,E+\Delta]$, a random state from this subspace is with high probability indistinguishable from the micro-canonical state. 

For a family of Hamiltonians for locally interacting quantum systems with increasing system size, if $\Delta$ is kept fixed and $E$ is chosen such that $R=[E,E+\Delta]$ is not too close to the boundaries of the spectrum of the Hamiltonian, then $d_R$ will typically grow exponentially with the system size $|\mathcal{V}|$. For a locally interacting system with a macroscopic number of constituents one would need to be able to distinguish an exponentially large number of measurement outcomes to have a realistic change of distinguishing a random state from a micro-canonical state. 

Typicality arguments have sometimes been called “unphysical”. The concept of typicality is complementary to other approaches towards the foundations of statistical mechanics and thermodynamics, such as ergodicity, the principle of maximum entropy, or postulating ensembles. However, especially with respect to the latter, typicality has at last one important advantage: It does not postulate that a certain ensemble yields a reasonable description of a certain physical situation, typicality shows, in a mathematically very well-defined way, when and why details do not matter. If most states anyway exhibit the same or very similar properties, then this does provide a heuristic, but convincing, argument in favor of the applicability of ensembles. It is hence an argument supporting a description of large systems with ensembles.

In order to be self-contained, we review aspects of pure state equilibration 
(for a review, see Ref.\ \cite{Gogolin_2016}).
The dynamics of finite dimensional quantum systems is manifestly recurrent and invariant under time reversal. This is in contrast to the thermodynamic behaviour observed in nature and the crux of the difficulty in deriving statistical mechanics and thermodynamics from quantum mechanics. 
This apparent contradiction 
can, however, be resolved to a large extent by considering the unitary time evolution of pure states as certain time dependent properties do indeed dynamically equilibrate. 

There are two notions of equilibration compatible with recurrence and time reversal invariance for finite dimensional quantum systems (equilibration on average and equilibration during intervals). Both of these notions capture the intuition that equilibration means that a quantity, after having been initialized at a non-equilibrium value, evolves towards some value and then stays close to it for an extended amount of time. What is referred to as equilibration is weaker than what one usually associated with the evolution towards thermal equilibrium. 

We can refer abstractly to time dependent properties of quantum systems by functions $f: \mathbb{R} \xrightarrow{} M$ that map time to some metric space $M$, for example $\mathbb{R}$ or $\mathcal{S}(\mathcal{H})$. The metric quantifies the values of the function at various times and how close a system is to its time average or ``equilibrium value''. 

We will be interested in properties that include the time evolution of one point functions. For this purpose, the notion of subsystem equilibration will be useful. In this case, the property is the time evolution of the state of the subsystem and the metric is the trace distance. We will focus on equilibration on average.  

\begin{dfn}
    [Equilibration on average] A time-dependent property equilibrates on average if its value is for most times during the evolution close to some equilibrium value. 
\end{dfn}

Equilibration on average, especially for expectation values of observables as well as for reduced states of small subsystems of large quantum systems is a provably generic feature. This implies that the equilibrium property will, for most times during the evolution, be close to the time average. This allows for a reasonable definition of an equilibrium state which is, for example, the time averaged state. An unfortunate short-coming is that on its own, equilibration on average does not have any implications on the time scales for which the equilibrium value is reached, after a system has started in an out of equilibrium situation. While these time scales can be bounded, they have a limited physical relevance. 

Proofs on equilibration require a few main ingredients. The first is the large dimension of the Hilbert space of most many body systems. The dimension of the Hilbert space of composite systems grows exponentially with the number of constituents. The actual object of concern is the number of significantly occupied energy levels. For each $k \in [d']$ the occupation of the $k-$th energy level is defined as $p_k:= \Tr(\Pi_k \rho(0))$ where $d':= |\text{spec}(H)| \leq d = \dim(\mathcal{H})$ is the number of distinct such levels. We can use $\max_kp_k$, the occupation of the most occupied level, to quantify the number of significantly occupied energy levels. For this purpose, we have the effective dimension, denoted as $D_{\text{eff}}$, which may be defined as 
\begin{equation}
    D_{\text{eff}} := \frac{1}{\sum_{k=1}^{d'}p_k^2} \geq \frac{1}{\max_k p_k}.
\end{equation}
If the initial state is taken to be an energy eigenstate, the resulting effective dimension is one, while if it results from a uniform coherent superposition of $\Tilde{d}$ energy eigenstates to different energies is $\Tilde{d}$. This justifies the interpretation of $d^{\text{eff}}(\omega)$ as a measure of significantly occupied states. It is a reciprocal to a quantity that is known as the inverse participation ration and is the time average of the Loschmidt echo. Using the effective dimension instead of the occupation of the most occupied level can lead to tighter bounds - it cannot, however, be efficiently computed given a state and the Hamiltonian. 
There are many ways to argue why it is acceptable to restrict oneself to initial states that populate a large number of energy levels when trying to prove the emergence of thermodynamic behavior from unitary dynamics in closed systems. The first one is that initial states which occupy a small subspace of Hilbert spaces of larger systems will behave essentially like small quantum systems and are not expected to behave thermodynamically but instead genuinely show quantum behavior. We will take the perspective of a measure concentration which guarantees that uniformly random pure states drawn from large subspaces of a Hilbert space will have, with extremely high probability, an effective dimension with respect to any fixed, sufficiently non-degenerate Hamiltonian that is comparable to the dimension of that subspace. If one is willing to assume that such states are physically natural initial states, this can justify the assumption of large effective dimension. 
For the purposes of equilibration, it will be sufficient that $\max_{k}'p_k$, the second largest energy level occupation is small. In a realistic physical situation that a system is cooled close to its ground state $\max_{k}'p_k$ can be orders of magnitude smaller than the inverse effective dimension. The physical intuition is that expectation values of observables of a system that are initialized in an energy eigenstates are already in equilibrium. What prevents equilibration on average are not macroscopic populations in a single energy level, but rather initial states which are coherent superpositions of a small number of energy eigenstates which will exhibit oscillations. 
A second ingredient to equilibration proofs is that of non-degenerate energy gaps which is called non-resonance. A Hamiltonian $H$ has non-degenerate energy gaps if for every $k,l,m,n \in [d']$
\begin{equation}
    E_k-E_l=E_m-E_n \longrightarrow (k=l \land m=n) \lor (k=m \land l=n),
\end{equation}
which is to say that every energy gap $E_k-E_l$ appears exactly once in the Hamiltonian spectrum. Older proofs exclude all Hamiltonians with degeneracies - the motivation for this condition is that the on average equilibration of the reduced state $\rho^S(t)$ of a small subsystem $S$ of a bipartite system with $\mathcal{V}=S \Dot{\cup} B$. If the Hamiltonian of the composite system is of the form 
\begin{equation}
    H = H_S + H_B,
\end{equation}
i.e., $S$ and $B$ are not coupled, then $\rho^S(t)$ will evolve unitarily and not equilibrate. Thus a condition is needed which excludes such  non-interacting Hamiltonians - the non-degeneracy of energy gaps is an elegant yet simple way to do this. More recent literature relaxes the conditions of non-degenerate energy gaps. One can restrict the maximal number of energy gaps in any energy interval of $\epsilon$ to 
\begin{equation}
\label{energylevel}
    N(\epsilon) := \sup_{E \in \mathbb{R}}|\{ (k,l) \in [d']^2: k \neq l \land E_k - E_l \in [E,E+\epsilon]\}|.
\end{equation}
Note that $N(0)$ is the number of degenerate energy gaps and a Hamiltonian $H$ satisfies the non-degenerate energy gap conditions iff $N(0)=1$. The definition allows one to prove an equilibration theorem that will work if a system has a small number of degenerate energy gaps. It has the advantage that it allows one to make statements about equilibration time. Equilibration on average can be guaranteed to happen on a time scale $T$ that is large enough such that $T \epsilon \gg 1$ where $\epsilon$ must be chosen small enough so that $N(\epsilon)$ is small compared to the number of significantly populated energy levels. 
\begin{thm}[Equilibration on average \cite{Gogolin_2016}]
\label{equibonav}
    Given a system with Hilbert space $\mathcal{H}$ and Hamiltonian $H \in \mathcal{O}(\mathcal{H})$ with spectral decomposition $H = \sum_{k=1}^{d'} E_k \Pi_k$. For $\rho(0) \in \mathcal{S}(\mathcal{H})$ the initial state of the system, let $\omega = \$_H(\rho(0))$ be the dephased state and define energy level occupations $p_k:=\Tr(\Pi_k \rho_0)$. Then for every $\epsilon, T >0$ it holds that: \\
    1. For any operator $A \in \mathcal{B}(\mathcal{H})$
    \begin{equation}
        \overline{(\langle A \rangle_{\rho(t)} - \langle A \rangle_\omega)^2}^T \leq \norm{A}_{\infty}^2 N(\epsilon)f(\epsilon T)g((p_k)_{k=1}^{d'}) \ .
    \end{equation}
    2. For every set $\mathcal{M}$ of POVMs
    \begin{equation}
        \overline{\mathcal{D}_{\mathcal{M}}(\rho(t), \omega)}^T \leq h(\mathcal{M}) (N(\epsilon)f(\epsilon T)g((p_k)_{k=1}^{d'}))^{1/2} \ ,
    \end{equation}
    where $N(\epsilon)$ is defined in \ref{energylevel}, $f(\epsilon T):= 1+8 \log_2(d')/(\epsilon T)$,
    \begin{align}
        & g((p_k)_{k=1}^{d'})) := \min(\sum_{k=1}^{d'}p_k^2, 3 
        \max_{k}' p_k) \ , \\  
        & h(\mathcal{M}):= \min( |\cup \mathcal{M}/4|, \dim(\mathcal{H}_{\text{supp}(\mathcal{M})})/2 ) \ ,
    \end{align}
    with $\max_{k}' p_k$ the second largest 
    element in $(p_k)_{k=1}^{d'}$, $\cup \mathcal{M}$ the set of all distinct POVM elements in $\mathcal{M}$, and $\text{supp}(\mathcal{M}):= \bigcup_{M \in \cup \mathcal{M}}\text{supp}(\mathcal{M})$, and
    and we define the finite time average $\overline{f}:= \frac{1}{T}\int_0^T f(t) dt$ and the infinite time average $\overline{f}:= \lim_{T \xrightarrow{} \infty} \overline{f}^T$.
\end{thm}
Let us elaborate upon the physical content of this theorem. The quantity $g((p_k)_{k=1}^{d'})$ is small, except if the intial state assigns large populations to a few but greater than one energy levels. For initial states with a reasonable energy uncertainty and large enough systems, it will be on the order of the reciprocal to the total number of distinct energy level $O(1/d')$. The quantity $h(\mathcal{M})$ measures the experimental capability of distinguishing quantum states and can be assumed to be smaller than $d'$. When all measurements in $\mathcal{M}$ have a support inside a small subsystem $S \subset \mathcal{V}$ it is bounded by $d_S/2$/ The theorem also implies an upper bound on $\overline{\mathcal{D}_{\mathcal{M}}(\rho(t), \omega)}^T$ and hence proves subsystem equilibration on average. For a fixed $H$ and $\epsilon >0$ we have $\lim_{T \xrightarrow{} \infty}f(\epsilon T)=1$, hence the theorem proves that for a wide class of reasonable initial states, equilibrium on average of all sufficiently small subsystems and apparent equilibration on average of the state of the full system under realistic restrictions on the number of different measurements can be preformed. 
The timescales on which equilibration can be reached is also something analyzable. The product $N(\epsilon)f(\epsilon T)$ which is lower bounded by one, will be typically close to one if $T$ is comparable to $d'^2$, i.e., to the total number of energy gaps, and will otherwise be on the order of $\Omega(d'^2/T)$ for smaller $T$. So even under the reasonable assumption that $g((p_k)_{k=1}^{d'})$  is of the order or $O(1/d')$, equilibration of a subsystem $S$ can only be guaranteed after a time $T$ which is roughly on the order of $\Omega(d_s^2 d')$. Both $d'$ and $d_S$ typically grow exponentially with the size of the composite system and the subsystem $S$ respectively. Times of order $\Omega(d_s^2 d')$ are unphysical for most systems of moderate size. 
\vspace{10pt}

\section{Fluctuations of the RTN norm and large deviation bounds}
\label{APP:FLUCTUATIONS}

Here we show that relative fluctuations of the RTN state norm $\braket{\psi}{\psi}/\mathbb{E}[\braket{\psi}{\psi}]$ disappear at large dimension, justifying the on-average normalization we employed in the \eqref{EQ_EFF_DIM}.
By applying Markov's inequality, we find that
\begin{align}
    \mathbbm{P} \left[\left( \frac{\braket{\psi}{\psi}}{\mathbb{E}[\braket{\psi}{\psi}]} - 1 \right)^2 \geq \varepsilon\right] 
    &\leq \frac{1}{\varepsilon}\, \mathbb{E}\left[ \left( \frac{\braket{\psi}{\psi}}{\mathbb{E}[\braket{\psi}{\psi}]} - 1 \right)^2 \right] \nonumber\\
    &= \frac{1}{\varepsilon}\, \left(\frac{\mathbb{E}[(\braket{\psi}{\psi})^2]}{(\mathbb{E}[\braket{\psi}{\psi}])^2} - 1 \right) \ .
\end{align}
We have already computed $\mathbb{E}[\braket{\psi}{\psi}]$ in Eq.\ \eqref{EQ_G_NORM}; assuming the physical dimension $a$ and bond dimension $b$ to be constant, we can write it concisely as  
\begin{align}
    \mathbb{E} [ \braket{\psi} ] = \frac{a^n\, b^{n_\text{int}}}{\prod_{k=1}^{n_\text{v}} q_k} \ .
\end{align}
We can then determine $\mathbb{E}[(\braket{\psi}{\psi})^2]$ via the fourth-order Weingarten formula \eqref{EQ_WEINGARTEN_STATE_4}. Again, we start with the illustrative example \eqref{EQ_TWO_TENSOR_EX}, for which we compute
\begin{widetext}
\begin{align}
    \label{normalization}
    \raisebox{-0.5\height}{\includegraphics[width=0.12\linewidth]{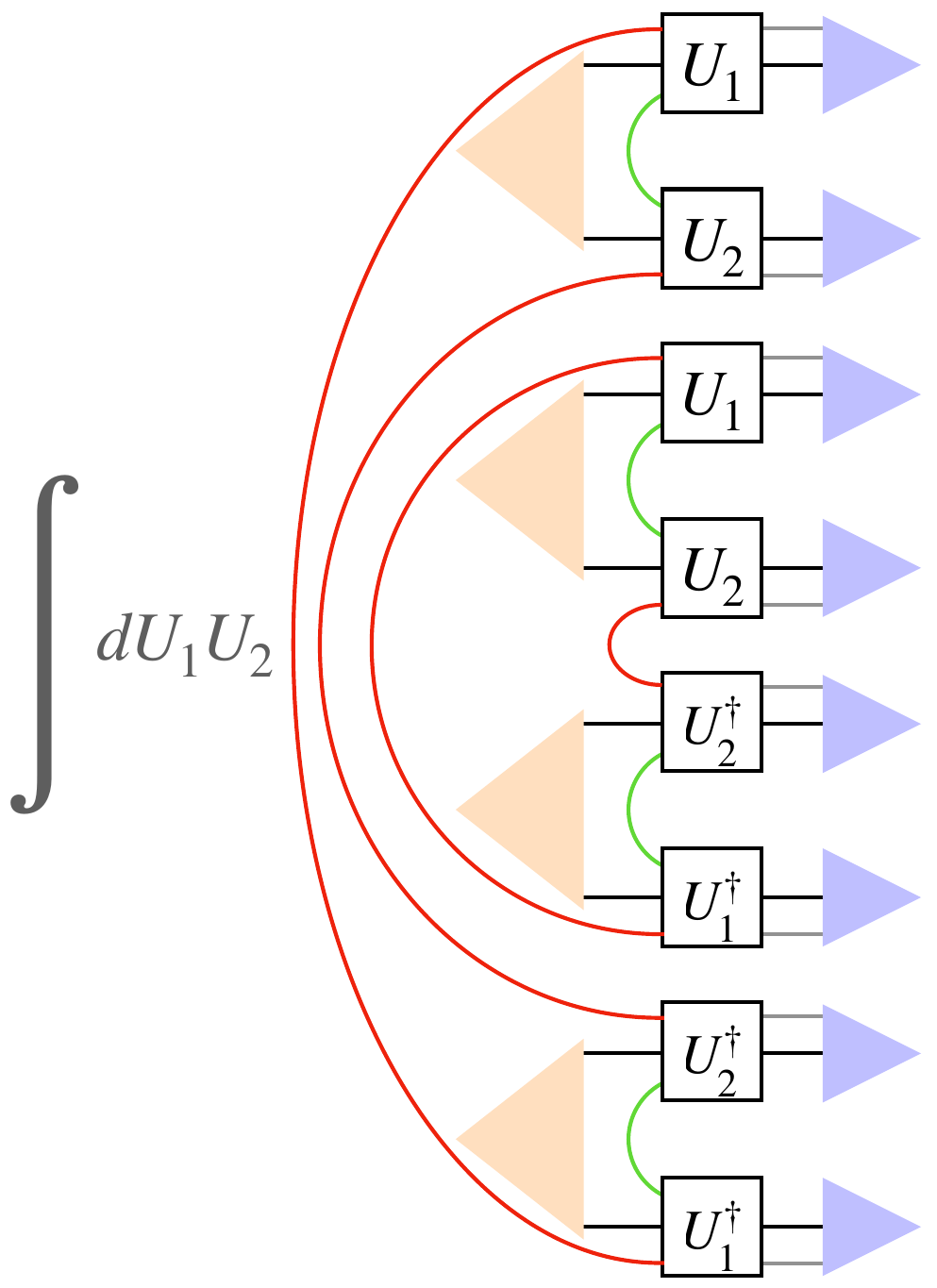}}
    &\mathrel{\vcenter{\hbox{$=$}}}
    \raisebox{-0.5\height}{\includegraphics[width=0.5\linewidth]{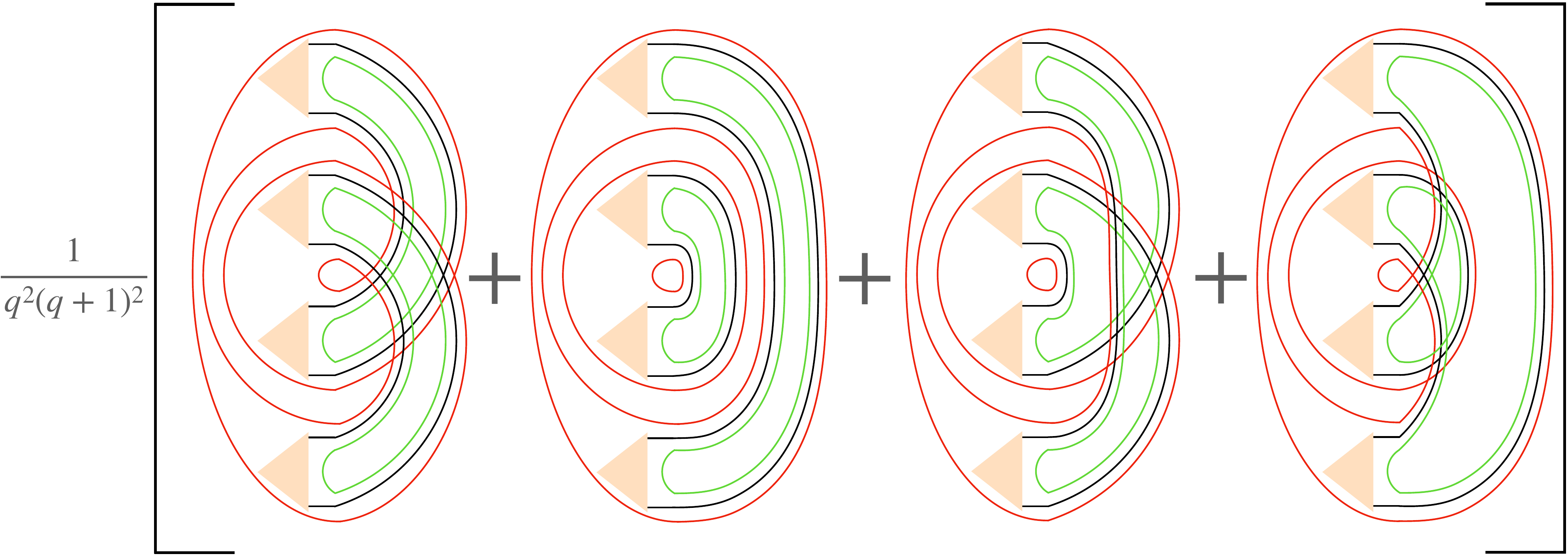}} \nonumber \\
    &\mathrel{\vcenter{\hbox{$=$}}}
    \raisebox{-0.5\height}
    {\includegraphics[width=0.5\linewidth]{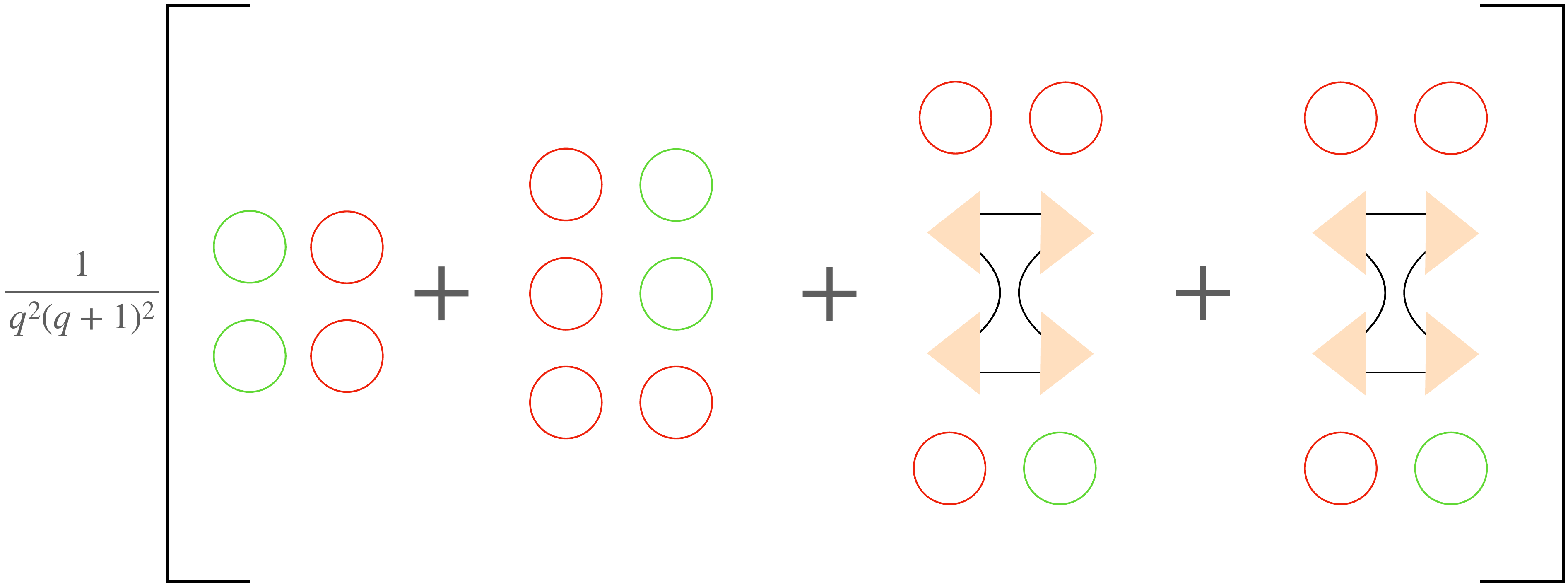}}\\
    &\mathrel{\vcenter{\hbox{$=$}}}
    \frac{a^4b^2+a^2b^2+2a^3b\tr_1[\rho_1^2]}{q^2(q+1)^2} \leq \frac{a^4b^2+a^2b^2+2a^3b}{q^2(q+1)^2}.\nonumber
\end{align}
\end{widetext}
The upper bound again follows from evaluating the connections between copies of the bulk state $\rho=\ketbra{\Phi}{\Phi}$ that produce $\tr_1 \rho_1^2 \leq 1$.
For this example, together with the expression for $\mathbb{E}[\braket{\psi}{\psi}]$ from \eqref{EQ_EXP_NORM_EX}, we then end up with a rigorous large deviation bound
\begin{align}
\label{complicated}
     \mathbbm{P} \left[\left( \frac{\braket{\psi}{\psi}}{\mathbb{E}[\braket{\psi}{\psi}]} - 1 \right)^2 \geq \varepsilon\right] 
    &\leq \frac{1}{\varepsilon} \left(\frac{a^4 b^2 + a^2 b^2 + 2a^3 b}{a^4 b^2 (1 + \frac{1}{q})^2} -1 \right) \nonumber\\
    &\leq  \frac{a^2 b^2 + 2a^3 b}{\varepsilon\, a^4 b^2} \ ,
\end{align}
where $q=a^2 b$ is again the total dimension per tensor.

For a general RTN geometry, we can again use the identity/flip formalism to map the fourth-order Weingarten expression into an Ising model partition function.
We first write
\begin{align}
    &\mathbb{E}[ (\braket{\psi}{\psi})^2 ] = \nonumber\\
    &\quad \bra{\sigma}^{\otimes 2}\bra{1}^{\otimes n}_\text{phys} \prod_{\langle i,j \rangle \in E_\text{int}} \bra{\chi_{i,j}}^{\otimes 2} \, \prod_{k=1}^{n_\text{v}} \frac{\ket{1}_k + \ket{F}_k}{q_k(q_k+1)} \ ,
\end{align}
with the difference to \eqref{EQ_G_OVERLAP4} being an additional projection $\bra{1}_\text{ext}$ for each boundary leg, corresponding to the inner product $\braket{\psi}{\psi}$ in each of the two factors. Again we end up with a sum over all vertices being either in the $\ket{1}$ or $\ket{F}$ configuration, with each state vector now acting on both internal and external (physical) legs.
With the $k$-th tensor having $m_k$ internal and $l_k-m_k$ external legs, this implies a decomposition
\begin{equation}
    \ket{1,F}_k = \ket{1,F}_\text{int}^{\otimes m_k} \otimes \ket{1,F}_\text{phys}^{\otimes (l_k-m_k)} \otimes \ket{1,F}_\text{log} \ ,
\end{equation}
with the inner products (omitting the logical states)
\begin{align}
    \braket{1}{1}_\text{int} &=\braket{F}{F}_\text{int}= b^2 \ , & \braket{1}{F}_\text{int} &= b \ , \\
    \braket{1}{1}_\text{ext} &=\braket{F}{F}_\text{phys} = a^2 \ , & \braket{1}{F}_\text{phys} &= a \ .
\end{align}
This brings us to the full expression
\begin{align}
    &\frac{\mathbb{E} [ (\braket{\psi}{\psi})^2 ]}{(\mathbb{E} [ \braket{\psi} ])^2} 
    = \frac{\bra{\Phi}\bra{1}_\text{phys}^{\otimes n}}{a^{2 n}}\prod_{\langle i,j \rangle \in E_\text{int}} \frac{\bra{\chi_{i,j}}^{\otimes 2}}{b_{\langle i,j \rangle}^2} \, \prod_{k=1}^{n_\text{v}} \frac{\ket{1}_k + \ket{F}_k}{1 +  \frac{1}{q_k}} \nonumber\\
    &=\prod_{k=1}^{n_\text{v}} \frac{1}{1 +  \frac{1}{q_k}} 
    \quad\sum_{\sigma} \tr_a[\rho_a^2]
    \prod_{\langle i,j \rangle \in E_\text{int}} b^{\frac{\sigma_i \sigma_j -1}{2}}  \prod_{\langle i,\emptyset \rangle \in E_\text{ext}} a^{\frac{\sigma_i-1}{2}}  \nonumber\\
    &\leq \prod_{k=1}^{n_\text{v}} \frac{1}{1 +  \frac{1}{q_k}} 
    \underbrace{\quad\sum_{\sigma}\prod_{\langle i,j \rangle \in E_\text{int}} b^{\frac{\sigma_i \sigma_j -1}{2}}  \prod_{\langle i,\emptyset \rangle \in E_\text{ext}} a^{\frac{\sigma_i-1}{2}} }_{Z_{\log\sqrt{b},\log\sqrt{a}}} \ ,
\end{align}
where $\langle i,\emptyset \rangle$ denotes an external edge connecting vertex $i$ to a physical site, and the sum runs over $n_\text{v}$ classical spins $\sigma_k = \pm1$.
This result can be expressed in terms of a modified Ising partition function with an external magnetic field $h_k$ on all external sites $k$,
\begin{align}
    Z_{J,h} &= \frac{1}{N} \sum_{\sigma}  e^{ \sum_{\langle i,j \rangle \in E_\text{int}} J_{\langle i,j\rangle}\sigma_i \sigma_j + \sum_{\langle i,\emptyset \rangle \in E_\text{ext}} h_i \sigma_i } \ , \\
    N &= e^{ \sum_{\langle i,j \rangle \in E_\text{int}} J_{\langle i,j\rangle} + \sum_{\langle i,\emptyset \rangle \in E_\text{ext}} h_i} \ .
\end{align}
Equivalently, it can be described by an Ising partition function on a new graph $G^\prime$ for which each boundary site is associated with an additional vertex with fixed spin $\sigma_k=+1$, and only the spins on the remaining vertices (i.e., those corresponding to tensors in the RTN) are allowed to vary. In this picture, all aligned spins lead to a factor $1$ per edge, while anti-aligned spins lead to either $1/a$ or $1/b$ depending on whether the edge connects to a boundary site or not.
In the original picture without boundary spins, we can expand $Z_{\log\sqrt{b},\log\sqrt{a}}$ at large dimensions $a$ and $b$ around the all-spin-up configuration, leading to
\begin{align}
    Z_{\log\sqrt{b},\log\sqrt{a}} = 1 + \sum_{k=1}^{n_v} \frac{1}{b^{m_k} a^{l_k-m_k}}+ \frac{1}{a^n}  +\dots \ , 
\end{align}
where we have included only corrections where a single spin $\sigma_k$ is the $-1$ configuration, as well as the term where \emph{all} spins are in the $-1$ configuration (the only term without $b$ dependence). Note that the fixed-spin boundary sites break the $\mathbb{Z}_2$ symmetry of the partition function under total spin flip, unlike $Z_{\log\sqrt{b}}$ considered in the main text.
Assuming that the tensor network forms a simply connected graph, each vertex connecting to a boundary site has at least one more internal edge, allowing us to bound
\begin{align}
    Z_{\log\sqrt{b},\log\sqrt{a}} = 1 + O\left( \frac{n_v}{b} + \frac{1}{a^n} \right) \ .
\end{align}
This finally leads us to a large deviation bound
\begin{align}
\label{EQ_NORM_DEV_BOUND_GEN}
     \mathbbm{P} \left[\left( \frac{\braket{\psi}{\psi}}{\mathbb{E}[\braket{\psi}{\psi}]} - 1 \right)^2 \geq \varepsilon\right] 
    &\leq \frac{1}{\varepsilon} \left( \frac{Z_{\log\sqrt{b},\log\sqrt{a}}}{\prod_{k=1}^{n_v} (1 + \frac{1}{q_k})} - 1\right) \nonumber\\
    &\leq O\left( \frac{n_v}{\varepsilon\,b} + \frac{1}{\varepsilon\,a^n} \right) \ .
\end{align}
This shows that relative norm fluctuations become negligible at large physical and bond dimensions $a$ and $b$ for general RTNs. Note that these bounds can be much tighter for specific RTN geometries: For example, an MPS/RTT geometry with closed boundary conditions will scale at most with $O({n}/{(\varepsilon a b^2)} + {1}/{(\varepsilon a^n}))$, so the large $a$ limit is sufficient to strongly suppress norm fluctuations even if $b$ remains small.

We now provide a physical and intuitive argument for why what we are doing is valid. Let the partition function $Z$ be a random variable, the annealed free energy is $F_{\text{annealed}}\approx \ln \mathbb{E}[Z]$ and the quenched free energy is $F_{\text{quenched}}\approx\mathbb{E}[\ln(Z)]$. By Jensen's inequality $\mathbb{E}[\ln Z] \leq \ln \mathbb{E}[Z]$. So the quenched free energy is lower bounded by the annealed free energy. We have equality if $Z$ is constant or there is no disorder. Otherwise, they differ by the cumulants, which in second order will be given by the variance of the random variable $\ln Z$. The mathematical analog of self-averaging in our case is measure concentration which justifies why we can move from averaging one to the other.

The deviations between sample-normalized and on-average-normalized expectation values can be quantified further.
For example, starting with the bound \eqref{EQ_NORM_DEV_BOUND_GEN}, we find that fluctuations of $\braket{\psi}{\psi}/\mathbb{E}[\braket{\psi}{\psi}]$ above $\varepsilon$ are suppressed at large $a$ and $b$. Let us write the right-hand side of Eq.\ \eqref{EQ_NORM_DEV_BOUND_GEN} as $\eta/\varepsilon$.
Considering the above large deviation bound, we can consider ``good events'' with deviation from normalization greater or equal than $\varepsilon$, and ``bad events'' with smaller variance.
One can then write the expectation value of inner products as a sum of the expectation value for the good and bad events and upper-bound each separately.
In this way, one finds for $\varepsilon<1$ as a rigorous bound
\begin{eqnarray}
    \mathbb{E}
    \left(\frac{\langle \psi|O|\psi\rangle}{
    \langle\psi|\psi\rangle
    }\right)\leq
    \frac{\mathbb{E} [\langle \psi|O|\psi\rangle]}{
    \mathbb{E}[\langle \psi|
    \psi\rangle]
    }\frac{1}{1-\sqrt{\varepsilon}}
    + \sqrt{\eta} \|O\| \ ,\quad
\end{eqnarray}
and via H{\"o}lder's inequality $\norm{fg}_1\leq \norm{f}_2\norm{g}_2$,
\begin{eqnarray}
    \mathbb{E}
    \left(\frac{\langle \psi|O|\psi\rangle}{
    \langle\psi|\psi\rangle
    }\right)\geq
    \frac{\mathbb{E} [\langle \psi|O|\psi\rangle]}{
    \mathbb{E}[\langle \psi|
    \psi\rangle]
    }\frac{1}{1+\sqrt{\varepsilon}}
    - \sqrt{\eta} \|O\| \ ,\quad
\end{eqnarray}
where $\|.\|$ denotes the operator norm, so the largest singular value of $O$. For practical purposes, this means that we can replace the above expectation values. 
In a similar way, one can upper bound the variance. One finds
\begin{align}
    \mathrm{Var}
\left(\frac{\langle \psi|O|\psi\rangle}{
    \langle\psi|\psi\rangle
    }\right)
    &\leq
    \frac{1}{(1-\sqrt{\varepsilon})^2}
\frac{\mathbbm{E}[
\langle\psi|O|\psi\rangle^2]
    }
    {\mathbbm{E}[
    \langle \psi|O|\psi\rangle]^2}\\
    \nonumber
    &+\sqrt{\eta}
\left(\mathbbm{E}\left[
\frac{\langle \psi|O|\psi\rangle^4}
{
\langle \psi|\psi\rangle^4}
\right]\right)^{1/2} ,
\end{align}
assuming the right-hand side is finite.
In the limit of increasing bond dimension, all higher cumulants will similarly concentrate. This means that the distribution is highly peaked. This then justifies why it is reasonable to replace the objects we are averaging over in the way we have done. Similar arguments can be made for the 2-R\'enyi entropy that actually features a Lifshitz constant upper-bounded by a constant. It is important to stress that this discussion is added here to render the discussion of the main text rigorous, but these arguments are interesting in their own right and also apply to other published results in the literature.

\section{Partition function relations under vertex fusion}
\label{APP:PROOF1}

We can lower-bound the ratio between $\Delta Z$ and $Z^\prime$ by considering how it is composed from terms of the partition function $Z_{\setminus j,k}$ of the vertices (and the edges connecting them) other than $j$ and $k$. 
Denoting as $m_k \leq l_k$ the number of \emph{internal} legs connected to the $k$-th vertex, it can be expanded into
\begin{equation}
    Z_{\setminus j,k} = \sum_{w_j=0}^{m_j-1} \sum_{w_k=0}^{m_k-1} Z_{\setminus j,k}(w_j,w_k) \ ,
\end{equation}
where $Z_{\setminus j,k}(w_j,w_k)$ contains all terms with a specific \emph{spin weight} (number of $-1$ spins) on the vertices connected to either $j$ or $k$. Symmetry under a spin flip on all vertices (excluding $j$ and $k$) results in an equality $Z_{\setminus j,k}(w_j,w_k) = Z_{\setminus j,k}(m_j-1-w_j,m_k-1-w_k)$.
We can then write
\begin{align}
    &Z^\prime = 2 \sum_{w_j=0}^{m_j-1} \sum_{w_k=0}^{m_k-1}  \frac{Z_{\setminus j,k}(w_j,w_k)}{b^{w_j+w_k}}  \nonumber\\
    &=   2 \sum_{w_j=0}^{m_j-1} \sum_{w_k=0}^{\lceil \frac{m_k-1}{2} \rceil}  \left( \frac{Z_{\setminus j,k}(w_j,w_k)}{b^{w_j+w_k}} + \frac{Z_{\setminus j,k}(w_j,w_k)}{b^{m_j+m_k-2-w_j-w_k}} \right) 
    \end{align}
    and
\begin{align}    
    &\Delta Z = 2 \sum_{w_j=0}^{m_j-1} \sum_{w_k=0}^{m_k}  \frac{Z_{\setminus j,k}(w_j,w_k)}{b^{m_j-w_j+w_k}}  \nonumber\\
    &=   2 \sum_{w_j=0}^{m_j-1} \sum_{w_k=0}^{\lceil \frac{m_k-1}{2} \rceil}  \left( \frac{Z_{\setminus j,k}(w_j,w_k)}{b^{m_j-w_j+w_k}} + \frac{Z_{\setminus j,k}(w_j,w_k)}{b^{m_k+w_j-w_k}} \right)  \ .
\end{align}
Here, we have assumed, without loss of generality, that all bonds have the same dimension $b$.
The terms in each sum arises all configurations in which $\sigma_j=\sigma_k=1$ (for $Z^\prime$) and $\sigma_j=-1,\sigma_k=+1$ (for $\Delta Z$), both of which are counting twice to account for the configuration of equal weight where all spins are flipped. In each second step, we then applied the spin flip symmetry for $Z_{\setminus j,k}$. Without knowledge of the remaining graph, the terms $Z_{\setminus j,k}(w_j,w_k)$ cannot be determined further; however, we can bound the minimum relative weight between each term pair of terms in $\Delta Z$ and $Z^\prime$ by the ratio of the $w_j=w_k=0$ terms:
\begin{align}
\frac{\frac{1}{b^{m_j+m_k-2-w_j-w_k}} + \frac{1}{b^{m_k+w_j-w_k}}}{\frac{1}{b^{w_j+w_k}} + \frac{1}{b^{m_k+w_j-w_k}}} &= \frac{b^{m_j+2w_k} + b^{m_k + 2w_j}}{b^2 + b^{m_j+m_k}} \nonumber\\
&\geq \frac{b^{m_j} + b^{m_k}}{b^2 + b^{m_j+m_k}} \ .
\end{align}
This means that the ratio $\Delta Z/Z^\prime$ can be lower-bounded in its entirety as
\begin{align}
\frac{\Delta Z}{Z^\prime} \geq \frac{b^{m_j} + b^{m_k}}{b^2 + b^{m_j+m_k}} \ .
\end{align}
Repeating the steps above for arbitrary dimension on any bond leads to the more general form
\begin{align}
\label{EQ_Z_RELATIONS_GENERAL}
\frac{\Delta Z}{Z^\prime} \geq \frac{q_j + q_k}{b_{\langle j,k \rangle}^2 + q_j\, q_k} \ .
\end{align}

\section{Partition function bounds for hyperbolic graphs}
\label{APP:PROOF2}

\begin{figure*}[ht]
    \centering
    \includegraphics[width=0.9\linewidth]{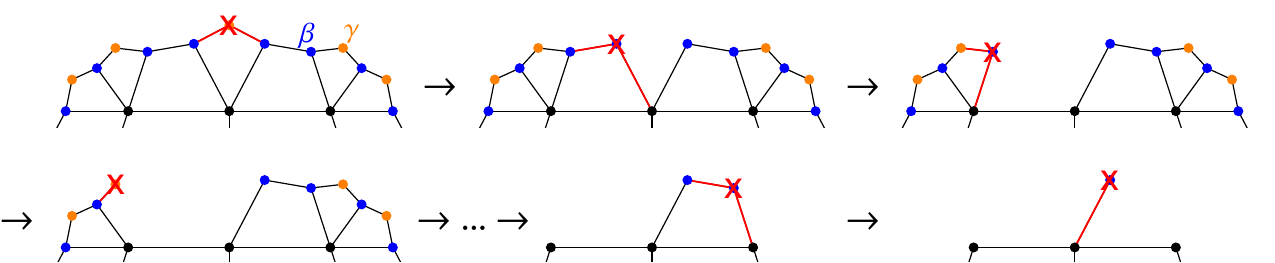}
    \caption{Iteratively removing vertices from a layer in the $\{5,4\}$ tiling. The vertices in the outer layer are color-coded according to their type ($\beta$ or $\gamma$), with a red cross denoting the vertex removed in each step and its connected edges shaded in red. We see that removing an $\beta$ or $\gamma$ vertex generally involves removing two or one edges, respectively, with the exception of the first (two instead of one) and last removed vertex (one instead of two).
    }
    \label{fig:vertex-removal}
\end{figure*}

Here we consider the Ising partition function $Z^{\{p,q\}}$ for hyperbolic graphs (i.e., $\frac{1}{p}+\frac{1}{q}<\frac{1}{2}$) in the regime $p,q>3$. Following the notation of Ref.\ \cite{Jahn:2025manylogical}, such graphs (or \emph{tilings} of the hyperbolic plane) can be built in layers of vertices through a process of \emph{vertex inflation} \cite{Boyle:2018uiv,Jahn:2019mbb} in which the vertices come in two types $\beta$ and $\gamma$. They are distinguished by how many legs of each vertex are connected to the next (outer) layer, with $\beta$ vertices having $p-3$ such legs and $\gamma$ vertices $p-2$ (another type $\alpha$ with $p=4$ legs appears only in $q=3$ tilings).
In each layer, the $\beta$ and $\gamma$ vertices alternate \emph{quasiperiodically}, with the sequence at one layer determining the sequence on the next via letter \emph{inflation rules} \cite{Boyle:2018uiv,Jahn:2019mbb}.
Specifically, the number $n_\beta^{(L)},n_\gamma^{(L)}$ of vertices of each type at the $L$-th inflation layer (the zeroth layer comprised of a single vertex) is given by
\begin{align}
    \begin{pmatrix}
        n_\beta^{(L+1)} \\
        n_\gamma^{(L+1)}
    \end{pmatrix}
    = M \begin{pmatrix}
        n_\beta^{(L)} \\
        n_\gamma^{(L)}
    \end{pmatrix} \ ,
\end{align}
where the \emph{substitution matrix} $M$ takes the form \cite{Jahn:2025manylogical}
\begin{equation}   
    M = 
    \begin{pmatrix}
        p-3 & p-2 \\
        (p-3)(q-3)-1 & (p-2)(q-3)-1  
    \end{pmatrix} \ .
\end{equation}
The ratio $n_\beta^{(L)}/n_\gamma^{(L)}$ quickly converges to a constant as $L$ is increased.
From the largest eigenvalue of $M$, one then finds
\begin{align}
\label{EQ_VERTEX_TYPE_RATIO_PQ}
    \frac{n_\beta^{(L)}}{n^{(L)}} &= \frac{p}{2} - 1 - \frac{1}{r} \ , \\
    \frac{n_\gamma^{(L)}}{n^{(L)}} &= 2 - \frac{p}{2} + \frac{1}{r} \ .
\end{align}
where $n^{(L)}=n_\beta^{(L)}+n_\gamma^{(L)}$ and $r$ is the rate of vertices over boundary sites defined in Eq.\ \eqref{EQ_RATE_HOLO_PQ}.

We now wish the run the inflation process in reverse, iteratively removing vertices and their connected edges and applying the bound \eqref{EQ_Z_BOUND_UPPER} in each step. Fig.\ \ref{fig:vertex-removal} shows this process for a single layer of vertices for the example of a $\{5,4\}$ tiling. Ignoring outer legs that are present when removing the outermost layer of the tiling (connected to the physical sites), which would only make the bound tighter, we see that removing a $\beta$ vertex involves removing two edges, while removing a $\gamma$ vertex involves removing one (for the first $\gamma$ and last $\beta$ vertex, the roles are reversed).
For a sufficiently large graph where the ratio between each type of vertex is well-approximated by \eqref{EQ_VERTEX_TYPE_RATIO_PQ}, we can invoke \eqref{EQ_Z_BOUND_UPPER} to bound
\begin{align}
    Z^{\{p,q\}} &\leq  \left( 1 + \frac{1}{b} \right)^{(2 - \frac{p}{2} + \frac{1}{r}) n_v}  \left( 1 + \frac{1}{b^2} \right)^{(\frac{p}{2} - 1 - \frac{1}{r}) n_v} \nonumber\\
&=  \left( \left( 1 + \frac{1}{b} \right)^{\frac{r(4-p)+1}{2}}  \left( 1 + \frac{1}{b^2} \right)^{\frac{r(p-2)-1}{2}} \right)^n \ ,
\end{align}
where we have written the last expression in terms of the number $n=n_v/r$ of physical sites.
In the $p,q>3$ regime, this upper bound takes its maximum value at $p=4.q=5$, so that we arrive at the general bound 
\begin{align}
    Z^{\{p,q\}} &\leq  \left( \left( 1+\frac{1}{b} \right)^{\frac{1}{2}} \left( 1+\frac{1}{b^2} \right)^{\frac{\sqrt{3}-1}{2}} \right)^n \ ,
\end{align}
for any $p,q>3$.

\section{Coulomb gas and eigenvalues of the CUE}
\label{APP:coulombgas}
Here, we provide a deeper background into the Coulomb gas formalism, largely following the discussion in Ref.\ \cite{forrester2010loggases}.
The canonical formalism of statistical mechanics applies to any mechanical system of $N$ particles that are free to move in a fixed domain $\Omega$, in equilibrium at temperature $T$. The
\emph{probability density function}
(p.d.f.)  
for the event that the particles are at positions $\mathbf{r}_1, \dots,  \mathbf{r}_N$ as 
\begin{equation}
    \frac{1}{\hat{Z}_N}e^{-\beta U(\mathbf{r}_1, \dots,  \mathbf{r}_N)}
\end{equation}
where $U(\mathbf{r}_1, \dots,  \mathbf{r}_N)$ denotes the total potential energy of the system and $\beta$ is Boltzmann's constant. The normalization is given by 
\begin{equation}
    \hat{Z}_N = \int_\Omega d\mathbf{r}_1 \dots \int_\Omega d \mathbf{r}_N e^{-\beta (\mathbf{r}_1, \dots,  \mathbf{r}_N)}.
\end{equation}
The term $ e^{-\beta (\mathbf{r}_1, \dots,  \mathbf{r}_N)}$ is the Boltzmann  factor and $\hat{Z}_N/N! := Z_N$ is the canonical partition function. 
For the log-potential Coulomb systems, the potential energy $U$ is calculated according to two dimensional electrostatics, and $\Omega$ must be one or two dimensional. The particles can be thought of as infinitely long parallel charged lines which are perpendicular to the confining domain. In a vacuum the electrostatic potential $\Phi$ at a point $\mathbf{r}=(x,y)$ due to a two-dimensional unit charge at $\mathbf{r}'=(x',y')$ is given by the solution of the Poisson equation 
\begin{equation}
\label{poisson}
\nabla^2_{\mathbf{r}}\Phi(\mathbf{r}, \mathbf{r}') = -2\pi \delta(\mathbf{r}-\mathbf{r}')
\end{equation}
where \begin{equation}
    \nabla^2_{\mathbf{r}}:= \frac{\partial^2}{\partial x^2} + \frac{\partial^2}{\partial y^2}.
\end{equation}
The solution of 
the Poisson equation 
is 
\begin{equation}
\label{soln1}
    \Phi(\mathbf{r},\mathbf{r}')=-\log(|\mathbf{r}-\mathbf{r}'|/l)
\end{equation}
where $l$ is an arbitrary length scale that will be set to one. A Coulomb gas is said to consist of one component if all $N$ particles are of like charge $q$ for example. To stop the particles from all repelling to the boundary, a neutralizing background charge density $-q\rho_b(\mathbf{r})$ is imposed, with the electro-neutrality condition $\int_\Omega \rho_b(\mathbf{r}) d\mathbf{r}=N$. The total potential energy $U$ therefore consists of the 
sum of the electrostatic energy of the particle-particle interaction 
\begin{equation}
    U_1 := -q^2 \sum_{1 \leq j<k \leq N} \log|\mathbf{r}_k - \mathbf{r}_j|,
\end{equation}
the particle-background interaction 
\begin{align}
     U_2&:= q^2 \sum_{j=1}^N V(\mathbf{r}_j), \\
     V(\mathbf{r}_j) &:= \int_\Omega \log|\mathbf{r}-\mathbf{r}_j|\rho_b(\mathbf{r})d\mathbf{r}
\end{align}
and the background-background interaction 
\begin{align}
    U_3 &:= \frac{-q^2}{2} \int_\Omega d\mathbf{r} \rho_b(\mathbf{r}) \int_\Omega d \mathbf{r} \rho_b(\mathbf{r}) \log|\mathbf{r}-\mathbf{r}_j| \\
    &=\frac{-q^2}{2} \int_\Omega \rho_b(\mathbf{r}')V(\mathbf{r}')d\mathbf{r}'
    \nonumber
\end{align}
the factor of $\frac{1}{2}$ in $U_3$ is included to compensate for the double counting of the potential energy implicit in the double integration. For this expression for $U$ we conclude that the Boltzmann factor of a one-component log-potential Coulomb system (log-gas) is of the form 
\begin{equation}
    e^{-\beta U_3} \prod_{l=1}^{N} e^{-\Gamma V(\mathbf{r}_l)} \prod_{1 \leq j<k \leq N}|\mathbf{r}-\mathbf{r}_j|^\Gamma
\end{equation}
where $\Gamma := q^2/k_BT$. For a given geometry and background density, the potentials $V(\mathbf{r})$ and $U_3$ can be explicitly evaluated. As an example, we single out the following.

\begin{thm}[Boltzmann  factor of a one-component log-potential Coulomb system]
\label{bolt}
    The Boltzmann  factor of a one-component log-potential Coulomb system of $N$ particles of charge $q=1$, confined to a circle of radius $R$ with a uniform neutralizing background, is given by 
    \begin{equation}
    \label{prod}
        R^{-N \beta/2} \prod_{1 \leq j < k \leq N} |e^{i \theta_k - e^{i \theta_j}}|^\beta
    \end{equation}
    where the position of each particle has been specified in polar coordinates. 
\end{thm}

\begin{proof}
    it is generally true that two points $\mathbf{r}$ and $\mathbf{r}'$ in the plane $|\mathbf{r} - \mathbf{r}'|= |Z - z'|$ are the corresponding points in the complex plane. Hence, if $\mathbf{r}$ and $\mathbf{r}'$ lie on a circle of radius $R$ with positions specified by polar coordinates, then $|\mathbf{r} - \mathbf{r}'| = R|e^{i \theta}-e^{i \theta'}|$. This formula gives the required expression for the product over pairs in (\ref{prod}) and allows the potential to be written as 
    \begin{align}
    V(\mathbf{r}) &= \frac{N}{2 \pi R} \int_0^{2\pi} \log|Re^{i \theta'} - R e^{i \theta}|R d\theta' \\
    &= N \log R + \frac{N}{2 \pi} \int_0^{2\pi} \log|e^{i\theta '}-1| d\theta'.
    \nonumber
    \end{align}
    The last integral vanishes and so $V(\mathbf{r})=N\log R$, using this result gives $U_3 = \frac{-q^2}{2}N^2 \log R$. Substituting these into (\ref{prod}) and noting $q=1, \Gamma=\beta$ gives the desired expression for the Boltzmann  factor. The Boltzmann  factor being proportional to the p.d.f.\  for the location of the particles, occurs in the definition of all statistical quantities associated with the equilibrium state. 
\end{proof}
We now consider the CUE ensemble where each eigenvalue can be written as $\lambda_j= e^{i \theta_j}$.
\begin{thm}[Eigenvalue probability density function of the CUE]
    The eigenvalue p.d.f.\ of the CUE is given by 
    \begin{equation}
        \frac{1}{C_{\beta,N}} \prod_{1 \leq j < k \leq N} |e^{i \theta_k - e^{i \theta_j}}|^2, \quad - \pi < \theta_1 \leq \pi.
    \end{equation}
\end{thm}
Comparing this with (\ref{bolt}), 
we see that the eigenvalue p.d.f.\ for the circular ensemble is directly proportional to the Boltzmann factor of the one-component log-potential Coulomb gas on a circle. An equivalent interpretation is that the log-gas is defined on a line with periodic boundary conditions. Suppose the line is in the $x-$direction and of length $L$. To specify a two dimensional Coulomb system in this setting, for the 
pair potential we seek a solution of the Poisson equation (\ref{poisson}) subject to the semi-periodic boundary condition $\Phi((x+L,y),(x',y'))=\Phi((x,y),(x',y'))$. The solution must depend on $x-x'$ and $y-y'$ and for $\mathbf{r} \sim \mathbf{r}'$, $\Phi(\mathbf{r}, \mathbf{r}') \sim - \log|\mathbf{r}-\mathbf{r}'|$ which is the solution of the Poisson equation in free boundary conditions with $l=1$. Complex analysis then allows us to assert that the real part of an analytic function satisfies Laplace's equation that 
\begin{align}
\label{pairpot}
    &\Phi(\mathbf{r},\mathbf{r}')= \nonumber \\ 
    &= - \log(|\sin(\pi(x-x' + i(y-y')/L))|(L/\pi)).
\end{align}
With the particles confined to the segment $[0,L]$ of the $x-$axis, this reduces to 
\begin{align}
    - \log(|\sin(\pi(x-x' + i(y-y')/L))|(L/\pi)) = \nonumber\\
    = -\log(|e^{2\pi i x/L}-e^{2 \pi i x'/L}|(L/2\pi))
\end{align}
thus revealing the equivalence to (\ref{soln1}) with $\mathbf{r}, \mathbf{r}'$ confined to a circle. 

Suppose for a log-gas system interacting via the pair potential (\ref{pairpot}) instead of the particles being confined to the line segment $[0,L]$ in the $x-$direction, they are confined to the full line in the $y-$direction. Up to an additive constant, the pair potential is then $-\log|\sinh(\pi(y-y')/L)|$, and if the particles are restrained from repelling to infinity by an attractive harmonic potential, the Boltzmann  factor is then of the form 
\begin{equation}
    \prod_{j=1}^N e^{-\beta c' y_j^2/2} \prod_{1 \leq j < k \leq N}|\sinh(\pi(y_k-y_j)/L)|^2, 
\end{equation}
for $-\infty < y_j < \infty$ where this corresponds to the partition function 
that occurs in Chern-Simons \cite{Giasemidis_2014}. 
    
\end{document}